\documentclass[11pt,uplatex]{article}
\usepackage{epsfig,multicol,bbm,amsmath,amssymb,amscd,mathrsfs,fancybox,amsthm,framed,enumerate,booktabs}
\usepackage{fancyhdr}

 \makeatletter
    
    \@addtoreset{equation}{section}
  \makeatother 
\usepackage[top=20truemm,bottom=20truemm,left=20truemm,right=20truemm]{geometry}

\newtheorem{Def}{Definition}[section]
\newtheorem{Thm}{Theorem}[section]
\newtheorem{Lem}{Lemma}[section]
\newtheorem{Prop}{Proposition}[section]
\newtheorem{Ass}{Assumption}[section]

\newcommand{\beq}{\begin{align}}
\newcommand{\eeq}{\end{align}}
\newcommand{\supp}{\mathrm{supp\,}}

\newcommand{\Expect}[1]{\left\langle{#1}\right\rangle}
\newcommand{\bvec}[1]{\boldsymbol{#1}}
\newcommand{\no}{\nonumber}

\newcommand{\Natural}{\mathbb{N}}

\newcommand{\Real}{\mathbb{R}}

\newcommand{\ad}{\mathrm{ad}}
\newcommand{\pdfrac}[2]{\frac{\partial #1}{\partial #2}}

\newcommand{\inprod}[2]{\langle #1| #2\rangle}
\newcommand{\norm}[1]{\left\| #1 \right\|} 

\newcommand{\ip}[2]{\left\langle#1,#2\right\rangle}
\newcommand{\q}{\quad}
\renewcommand{\H}{\mathcal{H}}
\newcommand{\K}{\mathcal{K}}
\newcommand{\V}{\mathcal{V}}
\newcommand{\F}{\mathcal{F}}
\DeclareMathOperator*{\slim}{s-lim}

\DeclareMathOperator*{\op}{\oplus}
\DeclareMathOperator*{\hop}{\hat{\oplus}}
\DeclareMathOperator*{\ot}{\otimes}
\DeclareMathOperator*{\hot}{\hat{\otimes}}
\DeclareMathOperator*{\wg}{\wedge}
\newcommand{\R}{\mathbb{R}}
\newcommand{\C}{\mathbb{C}}
\newcommand{\N}{\mathbb{N}}
\newcommand{\Z}{\mathbb{Z}}
\newcommand{\kk}{{\boldsymbol{k}}}
\newcommand{\pp}{{\boldsymbol{p}}}
\newcommand{\xx}{{\boldsymbol{x}}}
\newcommand{\XX}{{\boldsymbol{X}}}
\newcommand{\yy}{{\boldsymbol{y}}}

\newcommand{\bb}{\mathrm{b}}

\newcommand{\ee}{\mathbf{e}}

\renewcommand{\S}{\mathscr{S}}
\renewcommand{\rm}{\mathrm}

\renewcommand{\hat}{\widehat}

\newcommand{\CC}{\mathrm{(C)}}
\newcommand{\fin}{\mathrm{fin}}

\newcommand{\chiph}{\chi_\mathrm{ph}}
\newcommand{\pa}{\partial}
\newcommand{\Vp}{V_\mathrm{phys}}

\title{Eliminating unphysical photon components from Dirac-Maxwell Hamiltonian
quantized in the Lorenz gauge}
\author{Shinichiro Futakuchi\footnote{futakuchi@math.sci.hokudai.ac.jp} and Kouta Usui\footnote{kouta@math.sci.hokudai.ac.jp}\\
\small{Department of Mathematics, Hokkaido University, Sapporo, 060-0810, Japan.}}
\date{}

\begin{document}
\maketitle
\abstract{We study the Dirac-Maxwell model quantized in the Lorenz gauge. 
In this gauge, the space of quantum mechanical state vectors inevitably be an indefinite 
metric vector space so that the canonical commutation relation (CCR) is realized
in a Lorentz covariant manner. 
In order to obtain a physical subspace in which no negative norm state exists, the method first proposed by Gupta and Bleuler
is applied with mathematical rigor.
It is proved that a suitably defined
 physical subspace has a positive semi-definit metric, and naturally induces a physical Hilbert space with a positive definite metric.
 The original Dirac-Maxwell Hamiltonian naturally defines an induced Hamiltonian on the physical Hilbert space which is essentially self-adjoint.
}
\section{Introduction}
We consider a quantum system of $N$ Dirac particles under an external potential $V$ interacting with a quantized gauge field (so called Dirac-Maxwell model). 
If we apply the informal perturbation theory to this model, quantitative predictions are obtained
such as the Klein-Nishina formula for the cross section of the Compton scattering of an electron and a photon \cite{Nishijima1973a}, which
agrees with the experimental results very well. 
Hence, the Dirac-Maxwell model is expected to describe a certain class of realistic quantum phenomena 
and thus is worth the analysis with mathematical rigor, even though it may suffer from 
so called ``negative energy problem". The mathematically rigorous study of this model was initiated by Arai in Ref. \cite{MR1765584},
and several mathematical aspects of the model was analyzed so far (see, e.g.,
Refs. \cite{MR1981623}, \cite{MR2260374}, \cite{MR2810826}, \cite{MR2178588}, and \cite{MR2377946}). 

 The motivation of the present study is to treat this model in the Lorentz gauge in which the Lorenz covariance is manifest.
 In analyzing gauge theories such as quantum electrodynamics (QED) in a Lorentz covariant gauge, a difficulty always arises, since one  inevitably 
adopt ``an indefinite metric Hilbert space" as a space of all the state vectors (for instance, see Ref. \cite{MR509200}).
In such cases, we have to identify a positive definite subspace as a physical 
state vector space by eliminating unphysical photon modes with negative norms. 
The most general and elegant method to identify the physical subspace for non-abelian gauge theories quantized in a covariant gauge
was given by the celebrated work by Kugo and Ojima \cite{Kugo:1977zq,MR557151},
which is based on the BRST symmetry, the remnant gauge symmetry of the Lagrangian density after imposing some gauge fixing condition. 
The Kugo-Ojima formulation reduces to 
Nakanishi and Lautrup's $B$-field theory \cite{Lautrup:1967,Nakanishi:1973fu,Nakanishi:1974pr,Nakishi:1966di} in the case where the gauge field is abelinan and after integrating out the auxiliary Nakanishi-Lautrup's $B$-field, it is reduced to the condition first proposed by Gupta and Bleuler \cite{MR0038883,MR0036166}.

The Gupta and Bleuler condition says that a state vactor belongs to the phyisical subspace if and only if it has the   
vanishing expectation value of the operator $\partial_\mu A^{\mu}$, the four component divergence of the gauge field:
\begin{align}\label{GB-cond}
\inprod{\Psi}{\partial_\mu A^\mu (t,\xx) \Psi} = 0,
\end{align}
at every spacetime point $(t,\xx)\in\R^4$.
From the equations of motion $\Box A_\mu =- j_\mu$ (with ``mostly plus metric") and the current conservation equation $\partial_\mu j^\mu=0$, 
we heuristically find that $\partial_\mu A^\mu (t,\xx)$ satisfies the Klein-Gordon equation so that \eqref{GB-cond}
is written as
\begin{align}
[\partial_\mu A^\mu]^{+} (t,\xx) \Psi=0,
\end{align}
where $[\partial_\mu A^\mu]^{+} (t,\xx)$ denotes the positive frequency part of the free field $\partial_\mu A^\mu$.
However, in order to rigorously perform this procedure, one has to answer the following questions.
Firstly, how to identify $A(t,\xx)$ at a time $t\in \R$? Since the present state vector space is not an ordinary Hilbert space with a
positive definite metric, the Hamiltonian can not be defined as a self-adjoint operator in the ordinary sense.
Thus, it is far from trivial if there is a solution of quantum Heisenberg equations of motion.
Secondly, is it possible to identify the positive frequency part of the operator satisfying Klein-Gordon 
equation even in an indefinite metric space? The first problem is solved by the general construction method of
time evolution operator generated by a non-self-adjoint operator given by the authors \cite{FutakuchiUsui2013}, and as to the second 
one, the general definition of ``positive frequency part" of a quantum field satisfying Klein-Gordon equation is given in 
Ref. \cite{MR2533876}. 

Mathematically rigorous study of concrete models of QED in the Lorentz covariant gauge (see, for instance, Refs \cite{MR574710,MR2412280,MR2533876,IR-GB})
was only given for a solvable models as far as we know. However, the model treated here, the Dirac-Maxwell model, is \textit{not} solvable
in the sense that an explicit expression of the time-dependent gauge field is not easily found.
Thus, the problem has to be abstractly considered, not relying on the explicit expression of the time dependent gauge field
but only on the abstract existence theorem. In this paper, we establish the existence of a time evolution operator and give a definition of the 
``positive frequency part" of a free field safisfying Klein-Gordon equation in an abstract setup. 
Our definition of ``positive frequency part" given here is different from that given in Ref \cite{MR2533876}, but results in the same
consequence when applied to the concrete models. 
We then apply to the abstract theory to the concrete Dirac-Maxwell model and identify the physical Hilbert space.
We also prove that the original Dirac-Maxwell Hamiltonian naturally defines a self-adjoint ``physical Hamiltoninan" on the Hilbert space,
which is essentially equivalent to the Dirac-Maxwell Hamiltonian in the Coulomb gauge discussed in Ref. \cite{MR1765584}. 

The mathematical tools developed here would have some interests in its own right. The time evolution operator 
generated by \textit{unbounded, non-self-adjoint} Hamiltonians has been constructed in Ref. \cite{FutakuchiUsui2013}. 
In this paper, we further develop the theory in several aspects.
Firstly, we define a general class of operators, which we will call $\mathcal{C}_n$-class operators, and prove that
a $\mathcal{C}_n$-class operator $B$ has a time evolution $B(t)$ for $t\in\R$ which solves the Heisenberg equation,
and $B(t)\xi$ is $n$-times strongly differentiable in $t$ for $\xi$ belonging to a dense subset $D'$.  
Moreover, the $n$-th derivative of $B(t)$ enjoys the natural expression in terms of a weak commutator defined in a suitable sense.
Secondly, we define a more restricted class of operators, which will be called $\mathcal{C}_\omega$-class operators,
and prove that $\mathcal{C}_\omega$-class operators has a time evolution which is analytic in $t\in\R$.
Furthermore, it would be interesting to see that the following Taylor expansion formula \eqref{taylor} remains valid 
even for an unbounded, non-unitary time evolution.
Thirdly, the Klein-Gordon equation is analyzed in this abstract setup.
We generally define a ``free field"  as a solution of the generalized Klein-Gordon equation,
and then we prove that generalized Klein-Gordon equation is explicitly solved under some suitable conditions, even for an unbounded, non-unitary time evolution,
even in a vector space with an indefinite metric. 
From this explicit solution, positive
and negative frequency parts are defined unambiguously. 
Fourthly, abstract study with mathematical rigor makes it clear how and why the Gupta-Bleuler method works well and 
what is essentially responsible for the possibility of the method
is the electro-magnetic current conservation. 

%

The paper is organized as follows. In Section \ref{abstract}, after the brief 
review of the results obtained in \cite{FutakuchiUsui2013}, we discuss the abstract theory
of time evolution generated by non-self-adjoint Hamiltonians and give 
a definition of general free field to which the positive frequency part is defined.
In Section \ref{model}, we intoduce the Difinition of the model and how to mathematically deal with an 
indefinite metric space. 
In Section \ref{time-evolution}, we show that the Dirac-Maxwell Hamiltonian is essentially self-adjoint
with respect to the indefinite metric. We also show that the general theory developed so far is applicable to the Dirac-Maxwell Hamiltonian
and to study the time evolution of quantum fields. We derive the current conservation equation 
and the equation of motion.
In Section \ref{Gupta-Bleuler}, we identify the positive frequency part of the operator $\partial_\mu A^\mu(t,\xx)$
by using the abstract theory of generalized free field developed in Section \ref{time-evolution}.
Then, we define the physical subspace and the physical Hilbert space.
We also show that the original Hamiltonian naturally defines the physical Hamiltonian
which is essentially self-adjoint on the physical Hilbert space. The relation to the Coulomb gauge Hamiltonian 
is discussed there. Finally, we show that if the ultraviolet cutoff function of the gauge field is infrared singular, 
then the physical subspace has to become trivial.

\section{Time-evolution Generated by Non-self-adjoint Hamiltoninans}\label{abstract}
One of the main obstacles to mathematical analysis in the Lorenz gauge is
that the existence of the time-evolution operator is
not trivial at all because we can not use the
usual functional calculus to define the time-evolution operator $e^{-itH}$ if the generator $H$ is not self-adjoint.
The existence of a time-evolution is established by the general theory which the authors developed 
in Ref. \cite{FutakuchiUsui2013}. We summarize and further extend the 
results obtained in Ref. \cite{FutakuchiUsui2013} in an abstract setup. 

\subsection{Existence of time-evolution}
Let $\mathcal{H}$ be a Complex Hilbert space and $\ip {\cdot} {\cdot}$ its inner product, and
$\norm{\cdot}$ its norm. The inner product is linear in the second variable. 
For a linear operator $ T $ in $ \mathcal{H} $, we denote its domain (resp. range) by $ D(T) $ (resp. $ R(T) $). We also denote the adjoint of $T$ by $ T^* $ and the closure by $ \overline{T} $ if these exist. 
For a self-adjoint operator $ T $,  $ E_T (\cdot ) $ denotes the spectral measure of $ T $.
The symbol $T\upharpoonright D$ denotes the restriction of a linear operator $T$ 
to the subspace $D$.

Let $H_0$ be a self-adjoint operator on $\mathcal{H}$. Suppose that there is a nonnegative self-adjoint operator $A$ 
which is strongly commuting with $H_0$. We use the notations
\begin{align}
V_L&:=E_A([0,L]),\quad L\ge 0, \\
D&:=\bigcup_{L\ge 0} V_L,\\
D'&=D\cap D(H_0).
\end{align}
Let us define a family of linear operators $\mathcal{C}_0$ as follows:
\begin{Def}
We say that a linear operator $B$ is in $\mathcal{C}_0$-class if $B$ satisfies
\begin{enumerate}[(i)]
\item $B$ is densely defined and closed.
\item $B$ and $B^*$ are $A^{1/2}$- bounded.
\item There is a constant $b>0$ such that $\xi\in V_L$ implies $B\xi$ and $B^*\xi$ belong to $V_{L+b}$.
\end{enumerate}
\end{Def}
The set of all $\mathcal{C}_0$-class operators is also denoted by the same symbol $\mathcal{C}_0$.
We remark that if $B$ is in $\mathcal{C}_0$, then so $B^*$ is.
We consider an operator 
\begin{align}
H=H_0+H_1
\end{align}
with $H_1\in\mathcal{C}_0$. The following Propositions summarize the results obtained in Ref. \cite{FutakuchiUsui2013}.
\begin{Prop}\label{main-thm1} For each $ t,t' \in \R , \, \xi \in D $, the series:
\begin{align}
U(t,t') \xi := \xi + (-i) \int _{t'} ^t d\tau _1 \, H_1(\tau _1) \xi + (-i)^2 \int _{t'} ^t d\tau _1 \int _{t'} ^{\tau _1} d\tau _2 \, H_1 (\tau _1) H_1 (\tau _2) \xi + \cdots 
\end{align}
converges absolutely, where each of integrals are strong integrals. 
Furthermore, $D\subset D(U(t,t') ^*)$.
 \end{Prop}

We discuss the existence of the dynamics generated by $H$.
Let 
\begin{align}
W(t) := e^{-itH_0} \overline{U(t,0)} , \q t\in \R .
\end{align}
Then, we have
\begin{Prop}\label{sch-existence} For each $ \xi \in  D' $, the vector valued functions $ t\mapsto \xi (t) := W(t )\xi $ and $ t\mapsto \xi^*(t) := W(t )^* \xi $ 
 are strongly differentiable in $t\in \R$. Moreover, the followings hold:
\begin{align}\label{sch30}
\frac{d}{dt} \xi (t) &= -iH \xi (t) = -iHW(t)\xi = -i W(t)H\xi , \\
\frac{d}{dt} \xi^* (t) & = -iH^* \xi^* (t)= -iH^*W(t)^*\xi = -i W(t)^*H^*\xi .
\end{align}
Further, $W(t)\xi$ and $W(t)^*\xi$ belong to $D(A^{k/2})\cap D(H_0)$ for all $k=0,1,2,\dots$.
\end{Prop}

The existence of a solution of the Heisenberg equation 
is ensured if $B\in \mathcal{C}_0$ in a weak sense.
\begin{Prop}\label{w-Hei} 
Let $B\in\mathcal{C}_0$. Then,
\begin{enumerate}[(i)]
\item $ D \subset D(W(-t) B W(t)) $.
\item The operator-valued function $ B(t) $ defined as 
\begin{align}
D(B(t)) := D , \q B(t) \xi & := W(-t) B W(t) \xi , \q \xi \in D, \q t \in \R ,
\end{align}
is a solution of the weak Heisenberg equation:
\begin{align}\label{wHei2}
\frac{d}{dt} \left\langle \eta , B(t) \xi \right\rangle = \left\langle (iH)^* \eta , B(t) \xi \right\rangle - \left\langle B(t) ^* \eta , iH \xi \right\rangle , \q  \xi , \eta \in  D' .
\end{align}
\end{enumerate}
\end{Prop}

In Ref. \cite{FutakuchiUsuiSelf-ad}, a new criterion to prove the essential self-adjointness is proposed
as an application of these results. We refer two results obtained there for later use.

\begin{Prop}\label{a-or-b}
 Let $ T $ be a symmetric operator in $ \H $. If there exists a dense subspace $ V $ such that for any $ \xi \in V $ the initial value problem 
\begin{align}
\frac{d}{dt} \xi (t) = -i T \xi (t) , \q \xi (0 ) = \xi ,
\end{align}
has a strong solution $ \R \ni t \mapsto \xi (t ) \in D(T) $, then, 
exactly one of the following (a) or (b) holds.
\begin{enumerate}[(a)] 
\item $ T $ has no self-adjoint extension.

\item $ T $ is essentially self-adjoint.
\end{enumerate}
\end{Prop}
By using this Proposition, one readily finds 
\begin{Prop}\label{s.a.}
Let $H_1$ be in $\mathcal{C}_0$ and symmetric. Then, $H$ is also symmetric and for the symmetric operator $H$, exactly one of the following (a) and (b) holds:
\begin{enumerate}[(a)]
\item $H$ has no self-adjoint extension.
\item $H$ is essentially self-adjoint.
\end{enumerate} 
\end{Prop}

\subsection{$N$-th derivatives and Taylor expansion}
In this section, we develop a general theory concerning $n$-times differentiability of the operator $B(t)=W(-t)BW(-t)$
on a suitable subspace. Here, we take a slightly different formulation from that of Ref. \cite{FutakuchiUsui2013}
so that the generalization to $n$-th differentiability is easier.

To begin with, we prove a simple property of $B(t)$, which is not explicitly stated in Ref. \cite{FutakuchiUsui2013}.  
\begin{Lem}\label{s-conti-c_0}
Let $B\in\mathcal{C}_0$. Then the mapping
\begin{align}
t\mapsto W(-t)BW(t)\xi 
\end{align}
is strongly continuous in $t\in\Real$ for all $\xi \in D$.
\end{Lem}
\begin{proof}
Define an operator $U_n(t,t')$ on $D$, following Ref. \cite{FutakuchiUsui2013}, as
\begin{align}
U_n(t,t')\xi := (-i)^n\int_{t'}^t ds_1 \dots \int_{t'}^{s_{n-1}}ds_n \,H_1(s_1)\dots H_n(s_n)\xi.
\end{align}
Then the estimate 
\begin{align}\label{basic-estimate}
\norm{U_n(t,t')\xi} \le \frac{C^n}{n!}  |t-t'|^n (L+(n-1)b+1)^{1/2}\dots (L+1)^{1/2}\norm{\xi},
\end{align}
holds for some $C>0$ and $b>0$ (\cite{FutakuchiUsui2013} Lemma 3.4). From this estimate and the assumption that $B$ is in $\mathcal{C}_0$,
it is straightforward to check the mapping
\begin{align}
t\mapsto e^{itH_0}U_n(-t,0)Be^{-itH_0}U_m(t,0)\xi,\q \xi\in D
\end{align}
is strongly continuous. Since $W(-t)BW(t)\xi$ can be expanded in
a series converging absolutely and locally uniformly in $t\in\R$:
\begin{align}
BW(t)\xi = \sum_{n,m=0}^{\infty} e^{itH_0}U_n(-t,0)Be^{-itH_0}U_m(t,0)\xi,
\end{align}
the limit function
\begin{align}
t\mapsto W(-t)BW(t)\xi,\q \xi\in D
\end{align}
is also strongly continuous.
\end{proof}

\begin{Def}\label{C_1}
We say an operator $B$ is in $\mathcal{C}_1$-class if it satisfies 
\begin{enumerate}[(i)]
\item $B$ is in $\mathcal{C}_0$-class.
\item There is an operator $C\in\mathcal{C}_0$ such that
\begin{align}\label{w-comm}
\ip{(iH)^*\xi}{B\eta} - \ip{B^*\xi}{iH\eta}=\ip{\xi}{C\eta}
\end{align}
for all $\xi,\eta\in D'$.
\end{enumerate}
\end{Def}
We remark that the above operator $C\in\mathcal{C}_0$ is not unique in general. But one finds
\begin{Lem}
Let $B\in\mathcal{C}_1$ and $C$ be an operator mentioned in Definition \ref{C_1}. Then the operator
\begin{align}
\ad(B):=\overline{C\upharpoonright D(A^{1/2})}
\end{align}
does not depend on the choice of $C$ and determined only by $B$.
\end{Lem} 
\begin{proof}
Suppose that two operators $C_1$ and $C_2$ fulfills the condition. Then for all $\xi\in D'$,
we have
\begin{align}
C_1\xi=C_2\xi.
\end{align}
Take arbitrary $\eta\in D(A^{1/2})$ and put $\eta_n=E_A([0,n])E_{H_0}([-n,n])\eta$.
Then clearly $\eta_n\in D'$ and 
\begin{align}
\eta_n\to\eta,\q (A+1)^{1/2}\eta_n\to(A+1)^{1/2}\eta,
\end{align}
as $n$ tends to infinity. Therefore, we have
\begin{align}
C_1\eta &= \lim_{n\to\infty} C_1(A+1)^{-1/2}(A+1)^{1/2}\eta_n \no\\
	&= \lim_{n\to\infty} C_2(A+1)^{-1/2}(A+1)^{1/2}\eta_n \no\\
	&=C_2 \eta,
\end{align}
since both $C_1(A+1)^{-1/2}$ and $C_2(A+1)^{-1/2}$ are bounded. Thus one obtains
\begin{align}
C_1\upharpoonright D(A^{1/2})=C_2\upharpoonright D(A^{1/2}).
\end{align}
Taking the closure of both sides proves the assertion.
\end{proof}
\begin{Lem}\label{extended-w-comm}
Let $B\in\mathcal{C}_1$. Then the relation 
\eqref{w-comm} remains valid for all $\xi,\eta\in D(H_0)\cap D(A^{1/2})$. 
\end{Lem}
\begin{proof}
Let $\xi,\eta\in D(H_0)\cap D(A^{1/2})$. Put 
\begin{align}
\xi_n:=E_{A}([0,n])\xi, \q \eta_n:=E_{A}([0,n])\eta.
\end{align}
Then, clearly $\xi_n,\eta_n\in D'$ and 
\begin{align}
\xi_n\to\xi, \q H_0\xi_n\to H_0\xi,\q B\xi_n\to B\xi
\end{align}
as $n$ tends to infinity for all $B\in\mathcal{C}_0$, and the same is true for $\eta$.
The relation \eqref{w-comm} holds for $\xi_n $ and $\eta_n$. Thus, by the limiting argument,
the assertion follows.
\end{proof}
The strong Heisenberg equation is satisfied if $B$ is in $\mathcal{C}_1$-class.
\begin{Thm}\label{C_1-equation}
For $B\in \mathcal{C}_1$ and $\xi \in D'$ the mapping
\begin{align}
\R\ni t\mapsto B(t)\xi=W(-t)BW(t)\xi\in\mathcal{H}
\end{align}
is strongly continuously differentiable in $t\in\R$ and satisfies the Heisenberg equation of motion
\begin{align}
\frac{d}{dt} B(t)\xi = W(-t)\ad (B) W(t) \xi.
\end{align}
\end{Thm}
\begin{proof}
By Propositions \ref{sch-existence} and
\ref{w-Hei}, and Lemma \ref{extended-w-comm}, we have 
\begin{align}
\frac{d}{dt} \left\langle \eta , B(t) \xi \right\rangle &= \left\langle (iH)^* \eta , B(t) \xi \right\rangle - \left\langle B(t) ^* \eta , iH \xi \right\rangle , \no\\
&=\left\langle (iH)^* W(-t)^*\eta , BW(t) \xi \right\rangle - \left\langle B ^*W(-t)^* \eta , iH W(t) \xi \right\rangle \no\\
&=\ip{W(-t)^*\eta}{\ad(B)W(t)\xi}.
\end{align}
Since $\ad(B)$ is in $\mathcal{C}_0$, $\ad(B)W(t)\xi\in D(W(-t))$ and one gets
\begin{align}
\frac{d}{dt} \ip{\eta} {B(t)\xi} = \ip{\eta}{W(-t)\ad(B)W(t)\xi}.
\end{align}
Hence we have 
\begin{align}
\ip{\eta} {B(t)\xi} &= \ip{\eta}{\ad(B)\xi} + \int_0^t ds\, \ip{\eta}{W(-s)\ad(B)W(s)\xi} \no\\
	&=\ip{\eta}{\ad(B)\xi} +  \ip{\eta}{\int_0^t ds\,W(-s)\ad(B)W(s)\xi},
\end{align}
by Lemma \ref{s-conti-c_0}. Thus we obtain
\begin{align}
B(t)\xi=\ad(B)\xi + \int_0^t ds\,W(-s)\ad(B)W(s)\xi.
\end{align}
This equation shows that $B(t)\xi$ is continuously strongly differentiable and
\begin{align}
\frac{d}{dt}B(t)\xi = W(-t)\ad(B)W(t)\xi, \q \xi\in D'.
\end{align}
This completes the proof.
\end{proof}
One of the merits of the present formulation of the strong Heisenberg equation is that 
it is easy to extend for $n$-th differentiability. 
\begin{Def}
We define $\mathcal{C}_n$-class and $\ad^n(B)$ for $n=0,1,\dots$ inductively. That is, we say that an operator $B$ is in $\mathcal{C}_n$-class if $B$ is in $\mathcal{C}_{n-1}$-class and
$\ad(B)$ is in $\mathcal{C}_{n-1}$-class. For $B\in\mathcal{C}_n$, we write
\begin{align}
\ad^n(B):=\ad(\ad^{n-1}(B)),\q n=1,2,\dots.
\end{align}
It is clear that $\mathcal{C}_n\subset\mathcal{C}_{n-1}$ for $n=1,2,\dots$ and that if $B\in\mathcal{C}_n$, then $\ad(B)\in\mathcal{C}_{n-1}$, $\ad^2(B)\in\mathcal{C}_{n-2}$, \dots
$\ad^n(B)\in\mathcal{C}_0$. We define $\ad^0(B):=B$. An operator $B$ is said to be in $\mathcal{C}_\infty$-class if $B$ is in $\mathcal{C}_n$ for all $n\in\Natural$. Namely,
\begin{align}
\mathcal{C}_\infty := \bigcap_{n=0}^\infty \mathcal{C}_n.
\end{align}
\end{Def}

The following Theorem is important and useful but the proof is almost trivial by induction.
\begin{Thm}\label{n-th-diff}
Let $B$ is in $\mathcal{C}_n$-class. Then, for all $\xi\in D'$, $B(t)\xi$ is $n$-times strongly continuously differentiable 
in $t\in\R$ and 
\begin{align}
\frac{d^k}{dt^k}B(t)\xi = W(-t)\ad^k(B)W(t)\xi,\q k=0,1,2,\dots,n.
\end{align} 
In particular, if $B\in\mathcal{C}_\infty$, then
\begin{align}
\frac{d^k}{dt^k}B(t)\xi = W(-t)\ad^k(B)W(t)\xi,\q k=0,1,2,\dots.
\end{align} 
\end{Thm}
From Theorem \ref{n-th-diff}, we immediately have
\begin{Thm}
Let $B\in\mathcal{C}_n$ and $\xi\in D'$. Then, there is a $\theta\in (0,1)$ such that 
\begin{align}\label{finite-Taylor}
B(t)\xi = \sum_{k=0}^{n-1} \frac{t^k}{k!} \ad^k(B) \xi + \frac{t^n}{n!}W(-\theta t)\ad^n(B)W(\theta t)\xi. 
\end{align}
\end{Thm}
To obtain a Taylor series expansion for $B(t)\xi$ for $\xi\in D'$, we need one more concept.
\begin{Def}\label{c-omega}
We say that an operator $B$ is in $\mathcal{C}_\omega$-class if 
\begin{enumerate}[(i)]
\item $B\in\mathcal{C}_\infty$,
\item The operator norm 
\begin{align}
a_n:=\norm{\ad^n(B)(A+1)^{-1/2}}
\end{align} 
satisfies 
\begin{align}
\lim_{n\to\infty}\frac{t^na_n}{n!}=0, \q t>0.
\end{align}
\item There exists some constant $b>0$ such that
for all $n\ge 0$, $\xi\in V_L$ implies that $\ad^n(B)\xi$ belongs to $V_{L+b}$.
\end{enumerate}
\end{Def}
We then arrive at the following simple result.

\begin{Thm}\label{Taylor}
Suppose that $B\in\mathcal{C}_\omega$. Then, for each $\xi\in D'$, $B(t)\xi$ has the
norm-converging power series expansion formula
\begin{align}\label{taylor}
B(t)\xi = \sum_{n=0}^\infty \frac{t^n}{n!}\ad^n(B) \xi, \q t\in\R.
\end{align}
\end{Thm}
\begin{proof}
By Theorem \ref{n-th-diff}, all we have to show is that the norm of the reminder term in \eqref{finite-Taylor} vanishes:
\begin{align}
\lim_{n\to\infty} \norm{\frac{t^n}{n!}W(-\theta t)\ad^n(B)W(\theta t)\xi} = 0.
\end{align}
Since $B$ is in $\mathcal{C}_\omega$-class, there is a constant $b'>0$, which is independent of $n$, such that
$\xi\in V_{L}$ implies $\ad^n(B)\xi \in V_{B+b'}$. Choose $b>0$ such that 
$\xi\in V_{L}$ implies $H_1\xi \in V_{B+b}$ and $H_1^*\xi \in V_{B+b}$.
Put 
\begin{align}
C:=\norm{H_1(A+1)^{-1/2}}
\end{align}
which is finite since $B\in\mathcal{C}_\omega$, and let $\xi\in V_L$ for some $L\ge 0$.
Note that $W(-\theta t)\ad^n(B)W(\theta t)\xi$ is expanded in the norm-converging series and is estimated as
\begin{align}
\norm{\frac{t^n}{n!}W(-\theta t)\ad^n(B)W(\theta t)\xi} 
&\le \frac{|t|^n}{n!}\sum_{k,l=0}^\infty \norm{U_k(0,\theta t) e^{i\theta tH_0}\ad^n(B)e^{-i\theta tH_0}U_l(\theta t, 0)\xi}\no\\
&\le \frac{|t|^na_n}{n!}\sum_{k,l=0}^\infty \frac{C^{k+l}\,|\theta t|^{k+l}}{k!\,l!} (L+(k+l-2)b+b'+1)^{1/2}\dots\times\no\\
&\q\q \times(L+(l-1)b+b'+1)^{1/2}(L+(l-1)b+1)^{1/2}\dots (L+1)^{1/2}\norm\xi\no\\
&\to 0,
\end{align}
as $n$ tends to infinity.
This completes the proof.
\end{proof}

\subsection{Generalized Klein-Gordon equation and free fields}\label{GFF}
In this section, we investigate the solution of Klein-Gordon equation in mathematically
general formulation which is suitable to the present context.
Let $\mathfrak{h}$ be a complex Hilbert space and $T$ be a nonnegative self-adjoint operator
on $\mathfrak{h}$. 
\begin{Def}
A mapping $\phi:\mathfrak{h}\to\mathcal{C}_0$ is said to be $T$-free field if and only if
for $f\in D(T^2)$, $\phi(f)$ belongs to $\mathcal{C}_2$-class, and $\phi(t,f):=W(-t)\phi(f)W(t)$ satisfies the differential equation:
\begin{align}\label{KG}
\frac{d^2}{dt^2}\phi(t,f)\xi-\phi(t,-T^2 f)\xi=0,\quad \xi\in D',
\end{align}
where the differentiation is the strong one.
\end{Def}

We call this equation \textit{generalized Klein-Gordon equation}, since in the case where $\mathfrak{h}=L^2(\Real^3)$ and
$T=\sqrt{-\Delta}$, the equation \eqref{KG} gives the ordinary Klein-Gordon equation for a quantum field $\phi$.
As in the ordinary case, generalized Klein-Gordon equation can be explicitly solved,
in spite of the fact that, in the present case, the time evolution is not generated by a self-adjoint Hamiltonian.

We denote 
\begin{align}
C^\infty(T):=\bigcap_{n=1}^\infty D(T^n).
\end{align}

\begin{Lem}\label{even-odd}
Let $\phi$ be a $T$-free field. Then, for all $f\in C^\infty(T)$, $\phi(f)\in\mathcal{C}_\infty$ and
\begin{align}
\ad^{2n}[\phi(f)]\xi &= \phi((-1)^nT^{2n} f) \xi, \label{even}\\
\ad^{2n+1}[\phi(f)]\xi &= \ad[\phi((-1)^nT^{2n} f)] \xi, \label{odd}
\end{align}
for all $\xi\in D'$ and $n\ge 0$.
\end{Lem}
\begin{proof}
By Theorem \ref{n-th-diff}, the generalized Klein-Gordon equation \eqref{KG} is equivalent to 
\begin{align}
W(-t)\ad^2[\phi(f)]W(t)= W(-t)\phi(-T^2f)W(t)
\end{align}
on $D'$. At $t=0$ we have in particular
\begin{align}
\ad^2[\phi(f)]=\phi(-T^2f)
\end{align}
on $D'$. Since $f\in C^\infty(T)$, the right hand side belongs to $\mathcal{C}_2$,
which implies that $\phi(f)\in\mathcal{C}_4$ and 
\begin{align}
\ad^3[\phi(f)]&=\ad[\phi(-T^2f)], \\
\ad^4[\phi(f)]&=\ad^2[\phi(-T^2f)]=\phi((-T^2)^2f) \label{2nd-step}
\end{align}
on $D'$. But again the right hand side of \eqref{2nd-step} belongs to $\mathcal{C}_2$,
we obtain $\phi(f)\in\mathcal{C}_6$ and 
\begin{align}
\ad^5[\phi(f)]&=\ad[\phi((-T^2)^2f)], \\
\ad^6[\phi(f)]&=\ad^2[\phi((-T^2)^2f)]=\phi((-T^2)^3f). \label{3rd-step}
\end{align}
By repeating this argument, one finds that $\phi(f)\in\mathcal{C}_\infty$ and
\eqref{even}, \eqref{odd} hold.
\end{proof}

To solve the generalized Klein-Gordon equation, we introduce a well-behaved $T$-free field:
\begin{Def}\label{analytic-free}
A $T$-free field $\phi(\cdot)$ is said to be analytic if 
\begin{enumerate}[(i)]
\item For all $f$ which belongs to the subspace
\begin{align}
\bigcup_{N\in\N} E_T\left(\left[\frac{1}{N},N\right]\right),
\end{align}
$\phi(f)$ is in $\mathcal{C}_\omega$ class.
\item For $f\in D(T)$, $\phi(f)\in\mathcal{C}_1$.
\item $f_n\to f$ implies  
\begin{align}
\phi(f_n)\xi\to\phi(f)\xi ,\q \xi\in D',
\end{align}
and $Tf_n\to Tf$ implies 
 \begin{align}
 \ad[\phi(f_n)]\xi \to \ad[\phi(f)]\xi,\q \xi\in D'.
\end{align}
\end{enumerate}
\end{Def}
\begin{Thm}\label{KG-sol}
Let $\phi$ be an analytic $T$-free field. Then, for all $f\in D(T)$, we find
\begin{align}\label{sol-of-KG}
\phi(t,f)=\phi((\cos tT)f)+\ad\left[\phi\left(\left(\frac{\sin tT}{T}\right)f\right)\right]
\end{align} 
on $D'$.
\end{Thm}
\begin{proof}
Take $f\in R(E_T([1/N,N]))$ for some $N\in\N$. Then, $\phi(f)\in\mathcal{C}_\omega$. 
Thus we have
\begin{align}\label{exp-1}
\phi(t,f)=\sum_{n=0}^\infty \frac{t^n}{n!}\ad^n (\phi(f))
\end{align}
on $D'$ by Theorem \ref{Taylor}.
By Lemma \ref{even-odd}, one has for all $n\ge 0$,
\begin{align}
\ad^{2n}[\phi(f)]&=\phi((-1)^n T^{2n} f), \\
\ad^{2n+1}[\phi(f)]&=\ad[\phi((-1)^n T^{2n} f)],
\end{align}
on $D'$.
Then the Taylor expansion \eqref{exp-1} becomes
\begin{align}
\phi(t,f)=\sum_{n=0}^{\infty} \frac{t^{2n}}{(2n)!}\phi((-1)^n T^{2n} f) +\sum_{n=0}^{\infty} \frac{t^{2n+1}}{(2n+1)!}\ad(\phi((-1)^n T^{2n+1}T^{-1} f)) 
\end{align}
on $D'$. Note that as $N$ tends to infinity
\begin{align}
\sum_{n=0}^{N} \frac{(-1)^n t^{2n}}{(2n)!}T^{2n} f  \to (\cos tT) f,\q T\left(\sum_{n=0}^{N} \frac{(-1)^nt^{2n+1}}{(2n+1)!}T^{2n+1} T^{-1} f \right) \to T\left(\frac{\sin tT}{T}f\right).
\end{align}
Thus, by the assumption (iii) in Definition \ref{analytic-free}, 
we conclude that 
\begin{align}
\phi(t,f)=\phi((\cos tT)f )+\ad\left[\phi\left(\left(\frac{\sin tT}{T}\right)f\right)\right]
\end{align}
on $D'$ for all $f\in\cup_{N\in\N}R(E_T([1/N,N]))$. Here, the operators
$\cos tT$ and $\sin tT /T$ are defined through functional calculus.

Take arbitrary $f\in D(T)$ and put 
\begin{align}
f_N=E_T\left(\left[\frac{1}{N},N\right]\right)f.
\end{align}
One sees that the operators $\cos tT$ and $T^{-1}\sin t T $ are both bounded and
that $f_N$ converges to $f$ as $N$ tends to infinity. 
Hence, by the limiting argument, we find that $\eqref{sol-of-KG}$ remains valid for all $f\in D(T)$.
\end{proof}

This Theorem \ref{KG-sol} enables us to define positive and negative frequency parts of $\phi$:
\begin{Def}
Let $\phi$ be an analytic $T$-free field. We define for $f\in D(T^{-1})$, on $D'$
\begin{align}
\phi^{+}(t,f)&:=\phi\left(\frac{e^{-itT}}{2}f\right)-\ad\left[\phi\left(\frac{e^{-itT}}{2iT}f\right)\right],\label{positive-fp}\\
\phi^{-}(t,f)&:=\phi\left(\frac{e^{itT}}{2}f\right)+\ad\left[\phi\left(\frac{e^{itT}}{2iT}f\right)\right]
\end{align} 
and call $\phi^{+}$ (resp. $\phi^{-}$) positive (resp. negative) frequency part of $\phi$.
\end{Def}
The operators in the parentheses are defined though the 
functional calculus of self-adjoint operator $T$. 
It is clear by definition that an analytic $T$-free field $\phi$ can be written as a sum of
its positive and negative frequency parts,
\begin{align}
\phi(t,f)=\phi^{+}(t,f)+\phi^-(t,f),\q t\in\R, \q f\in D(T^{-1}),
\end{align}
on the subspace $D'$.

\section{Definition of the Dirac-Maxwell Hamiltonian in the Lorenz Gauge}\label{model}
In this section, we introduce the Dirac-Maxwell Hamiltonian $ H_\rm{DM}(V,N) $ quantized in the Lorenz gauge.
This Hamiltonian describes 
a quantum system consisting of a Dirac particle under a potential $V$ and a gauge field minimally interacting
with each other.
We use the unit system in which the speed of light and $ \hbar $, the Planck constant
devided by $2\pi$, are set to be unity.

\subsection{Dirac particle sector}
Let us denote the mass and the charge of the Dirac particle by $ M>0 $ and $ q \in \Real $, respectively.   The Hilbert space of state vectors for the Dirac particle is taken to be 
\begin{align}
\H _\mathrm{D} := L^2 (\Real_{\xx} ^3; \C ^4),
\end{align}
the square integrable functions on $\Real_{\xx}^3 = \{ \xx = (x^1 , x^2 , x^3) \, | \, x^j \in \R , \, j=1,2,3 \} $
 into $\C^4$. The vector space $\R^3_\xx$ here represents
the position space of the Dirac particle. We sometimes omit the subscript $\xx$ and just write $\R^3$ instead of $\R^3_\xx$
when no confusion may occur.
The target space $\C^4$ realizes a representation of the four dimensional 
Clifford algebra accompanied by the four dimensional Minkowski vector space. The generators 
$\{\gamma^\mu\}_{\mu=0,1,2,3}$ 
satisfy the anti-commutation relations
\begin{align}
\{\gamma^\mu, \gamma^\nu\}:=\gamma^\mu\gamma^\nu+\gamma^\nu\gamma^\mu=-2\eta^{\mu\nu},\q \mu,\nu=0,1,2,3,
\end{align}
where the Minkowski metric tensor $\eta = (\eta_{\mu\nu} )$ is given by
\begin{align}
\eta= \left(\begin{matrix}
-1 & 0& 0& 0 \\
0 & 1 & 0& 0 \\
0& 0& 1 &0 \\
0&0&0& 1
\end{matrix}\right) .
\end{align}
We set $ \eta ^{-1} = (\eta ^{\mu \nu}) $, the inverse matrix of $ \eta $, Then we have $ \eta ^{\mu \nu} = \eta _{\mu \nu} , \; \mu , \nu =0,1,2,3 $.
We assume $\gamma^0$ to be Hermitian and $\gamma^j$'s ($j=1,2,3$) be anti-Hermitian.
We use the notations following Dirac:
\begin{align}
\beta:=\gamma^0 , \q \alpha^j := \gamma^0 \gamma^j,\q j=1,2,3.
\end{align}
Then, $\alpha^j$s and $\beta$ satisfy the anti-commutation relations
\begin{align}
& \{ \alpha ^i , \alpha ^j \} = 2 \delta ^{ij} , \q i,j =1,2,3, \\
& \{ \alpha ^j , \beta \} =0 , \q \beta ^2 =1, \q j=1,2,3,
\end{align}
where $ \delta _{ij} $ is the Kronecker delta.
The momentum operator of the Dirac particle is given by
\begin{align}
\mathbf{p} := (p_1 , p_2 , p_3) := (-i D_1 , -iD_2 , -iD_3)
\end{align}
with $ D_j $ being the generalized partial differential operator on $ L^2 (\Real ^3; \C ^4) $ 
with respect to the variable $ x^j $, the $ j $-th component of $ \bvec{x} = (x^1 , x^2 , x^3) \in \R ^3 $. 
We write in short
\[  { \boldsymbol \alpha } \cdot \bvec{p} := \sum _{j=1} ^3 \alpha ^j p_j .\]
The potential is represented by a $ 4\times 4 $ Hermitian matrix-valued function $ V $ on $ \R ^3_\xx$ with each matrix components 
being Borel measurable. Note that the function $V$ naturally defines a linear operator acting in $\H_\rm{D}$ and
we denote it by the same symbol $V$.
The Hamiltonian of the Dirac particle under the influence of this external potential $V$ is then given by the Dirac operator
\begin{align}
H_\mathrm{D} (V) := { \boldsymbol \alpha } \cdot \bvec{p} +M \beta +V 
\end{align}
acting in $\H_{\mathrm{D}}$,
with the domain $ D(H_\mathrm{D} (V)) := H^1 (\Real ^3 ; \C ^4) \cap D(V) $, where $ H^1 (\Real ^3; \C ^4 ) $
denotes the $ \C ^4 $-valued Sobolev space of order one. 
Let $C$ be the conjugation operator in $\mathcal{H}_\rm{D}$ defined by
\[ (Cf)(\xx)=f(\xx)^* ,\q f\in \H_\rm{D},\q x\in \R^3,\]
where $*$ means the usual complex conjugation.
By Pauli's lemma \cite{MR1219537}, there is a $4\times 4$ unitary matrix $U$ satisfying
\begin{align}
U^2&=1, \q UC=CU, \label{Pauli-matrix1}\\
U^{-1}\alpha^j U &= \overline{\alpha^j}, \q j=1,2,3, \q U^{-1}\beta U = -\beta, \label{Pauli-matrix2}
\end{align}
where for a matrix $A$, $\overline{A}$ denotes its complex-conjugated matrix
and $1$ the identity matrix.
We assume that the potential $V$ satisfies the following conditions :
\begin{Ass}\label{ass-V}
\begin{enumerate}[(I)]
\item Each matrix component of $V$ belongs to
\[ L^2_\rm{loc}(\R^3):=\left\{f:\R^3\to \C\, \Bigg|\, \text{Borel measurable and } \int_{|\xx|\le R}|f(\xx)|^2<\infty \,\text{ for all $R>0$.} \right\}.\]
\item $V$ is Charge-Parity (CP) invariant in the following sense:
\begin{align}
U^{-1}V(\xx)U = V (-\xx)^*, \q \rm{a.e.} \,x\in\R^3.
\end{align}
\item $H_\rm{D}(V)$ is essentially self-adjoint.
\end{enumerate}
\end{Ass}
\noindent Hereafter, we denote the closure of $H_\rm{D}(V)$, which is self-adjoint by Assumption
\ref{ass-V}, by the same symbol.
The important remark is that the Coulomb type potential 
\begin{align}\label{Cou}
V(\xx)=-\frac{Zq^2}{|\xx|} 
\end{align}
satisfies Assumption \ref{ass-V} provided that $Zq^2 <1/2$, or more concretely,
$Z\le 68$ if we put $q=e$, the elementary charge \cite{MR1219537}.

Suppose that there are $N$ Dirac particles in the external potential $V$. In this case, the Hilbert space should be
\begin{align}
\H_\rm{D}(N):=\wedge^N \H_\rm{D} := \ot_\rm{as}^N L^2(\Real^3;\C^4)= L^2_{\rm{as}}((\R^3\times\{1,2,3,4\})^{N}),
\end{align}
where $\ot_\rm{as}^N$ denotes the $N$-fold anti-symmetric tensor product.
The $a$-th component of 
\[ \XX = (\xx^1,l^1;\dots;\xx^N,l^N) \in (\R^{3}\times\{1,2,3,4\})^N \]
represents the position and the spinor of the $a$-th Dirac particle.
For notational simplicity, we denote the position-spinor space of one electron by
$\mathcal{X}=\R^3\times\{1,2,3,4\}$ in what follows. We regard $\mathcal{X}$ as
a topological space with the product topology of the ordinary one on $\R^3$ and 
the discrete one on $\{1,2,3,4\}$.
The $N$ particle Hamiltonian is then given by
\begin{align}
H_D(V,N)=\sum_{a=1}^N \big(1\otimes  \dots \otimes \stackrel{a\text{-th}}{H_D(V)}\otimes\dots\otimes 1\big),
\end{align}
which is written as 
\begin{align}
H_D(V,N)=\sum_{a=1}^N\big(\alpha^a\cdot \bvec p^a + \beta^a M^a + V^a \big).
\end{align}
with $\pp^a = (-iD^a_{1},-iD^a_{2},-iD^a_{3}) $ denoting the generalized differential operator with respect to the $a$-th coordinate, $\alpha^a$ and $\beta^a$ denoting the operators $ 1\otimes \dots \otimes \stackrel{a\text{-th}}{\alpha}\otimes\dots\otimes 1 $ and $ 1\otimes \dots \otimes \stackrel{a\text{-th}}{\beta}\otimes\dots\otimes 1 $ acting in $\ot^N \C^4$, respectively, 
$M^a$ being the mass of the $a$-th Dirac particle, 
and $V^a$ the matrix-valued multiplication operator in $\wedge^N \H_\rm{D}$
 by the matrix-valued function
$\mathcal{X}^N \ni (\xx^1,l^1;\dots;\xx^N,l^N)\mapsto V(\xx^a) $, acting as
\begin{align}
(V^a\Psi)(\xx^1,l^1;\dots,\xx^N,l^N)=\sum_{k^a}V_{l^ak^a}(\xx^a)\Psi(\xx^1,l^1;\dots; \xx^a,k^a;\dots ;\xx^N,l^N).
\end{align}

\subsection{Free Hamiltonian of gauge field} 
Next, we introduce the free gauge field Hamiltonian in the Lorenz gauge.
We adopt as the one-photon Hilbert space
\begin{align}
\H _\mathrm{ph} := L^2 (\R ^3 _\kk ; \C ^4) .
\end{align}
The above $ \R ^3 _\kk := \{ \mathbf{k} = (k^1,k^2,k^3) \, | \, k^j \in \R , \, j=1,2,3 \} $ represents the momentum space of photons, and the the target space $\C^4$ represents the
degrees of freedom coming from the polarization of photons. As is well known, there are four
polarizations, two of which are physical (or transverse) ones while the other two 
are unphysical, scalar and longitudinal ones. 
We often omit the subscript $ \kk $ in $ \R ^3 _\kk $, and just denote it by $\R^3$, when there is no danger of confusion.
The Hilbert space for the quantized gauge field in the Lorentz gauge is given by 
\begin{align}
\F _\mathrm{ph} :=  \op _{n=0} ^\infty \ot _\rm{s} ^n \H _\rm{ph} = \Big\{ \Psi = \{ \Psi ^{(n)} \} _{n=0} ^\infty \, \Big| \, \Psi ^{(n)} \in \ot _\rm{s} ^n \H _\rm{ph} , \, \, ||\Psi ||^2:=\sum _{n=0} ^\infty \| \Psi ^{(n)} \| ^2 <\infty \Big\} ,
\end{align}
the Boson Fock space over $ \H _\mathrm{ph} $, where $ \ot _\rm{s} ^n $ denotes the $ n $-fold symmetric tensor product with the convention $ \ot _\rm{s} ^0 \H _\rm{ph} := \C $.
Let $ \omega (\kk ) := |\kk | , \, \kk \in \R ^3 $, the energy of a photon with momentum $\kk\in
\Real^3$. The multiplication operator by the $4\times 4$ matrix-valued function 
\begin{align}
\kk\mapsto
\begin{pmatrix}
 \omega (\kk) & 0 & 0& 0\\
 0 &\omega(\kk) &0&0 \\
 0&0&\omega(\kk) &0\\
 0&0&0&\omega(\kk)
\end{pmatrix}
 \end{align}
acting in $\H_\rm{ph}$ is self-adjoint, and we also denote it by the same symbol $ \omega $. This operator 
$\omega$ is a one-photon Hamiltonian
in $\mathcal{H}_{\rm{ph}}$, and 
the free Hamiltonian (kinetic term) of the quantum gauge field is given by its second quantization
\begin{align}
H_\mathrm{ph} &:= \mathrm{d}\Gamma _\bb (\omega ) :=\op_{n=0}^\infty \overline{ \Big( \sum _{j=1} ^n 1 \otimes \dots \otimes  1 \otimes \stackrel{j\text{-th}}{\omega} \otimes 1 \otimes \dots \otimes 1 \Big) \upharpoonright \hot ^n D(\omega ) } ,
\end{align}
where $\hot$ denotes the algebraic tensor product.
The operator $ H_\mathrm{ph} $ is self-adjoint.

\subsection{Canonical commutation relations and indefinite metric}
The main obstacle for the mathematical treatment of the Lorentz covariant gauge is 
that in order to realize the canonical commutation relations in a Lorentz covariant manner,
we have to employ an indefinite metric vector space as a space of state vectors.
In the above Hilbert space $\H _\mathrm{ph}$ for photon fields, the ordinary positive definite inner
product 
\begin{align}
\langle F, G\rangle:=\sum_{\mu=0}^3 \int_{\Real^3} F^\mu(\kk)^*G^\mu(\kk) \,d\xx
\end{align}
for $F(\kk)=(F^0(\kk),\dots,F^3(\kk))$ and $G(\kk)=(G^0(\kk),\dots,G^3(\kk))$ in $\F _\mathrm{ph}$, \textit{plays
no physical role} and just define the Hilbert space topology on the vector space $\F _\mathrm{ph}$. 
 For instance, the probability amplitude that a state $\Psi\in\F _\mathrm{ph}$
is observed in a state $\Phi\in\F _\mathrm{ph}$ is not given by $\langle\Psi,\Phi\rangle$. 
The ``physical" inner product is given by the mapping
\begin{align}
\langle F| G\rangle:= \int_{\Real^3} \eta_{\mu\nu} F^\mu(\kk)^*G^\nu(\kk) \,d\xx,
\end{align} 
where the summation over $\mu,\nu= 0,1,2,3$ is understood (In what follows, we omit the summation
symbol whenever the summation is taken with respect to one upper and one lower Lorentz indicies).
The indefinite metric naturally induced on $\F_\rm{ph}$ from this indefinite metric on $\H_\rm{ph}$ is a physical one in the sense that $\langle\Psi | \Phi\rangle$ gives the probability amplitude that a state $\Psi\in\F _\mathrm{ph}$
is observed in a state $\Phi\in\F _\mathrm{ph}$.
Note that $\langle\Psi | \Psi\rangle$ may become negative and thus $\F_\rm{ph}$ contains a lot of unphysical state vectors with
``negative probability", which should be eliminated.
Note also that the self-adjointness which is to be required for the Hamiltonian has to be that with respect to the indefinite metric $\langle\cdot|\cdot\rangle$,
in stead of the ordinary one with respect to the unphysical inner product $\Expect{\cdot,\cdot}$.  

We make precise the statement that a linear operator 
is self-adjoint with respect to the indefinite metric $\langle\cdot|\cdot\rangle$.
The metric tensor $\eta=(\eta_{\mu\nu})$ is considered to be a linear operator 
on $\H_\rm{ph}$ by
\[ (\eta F)^\nu (\kk)=\eta_{\mu\nu} F^\mu(\kk),\q F=(F^\mu)_{\mu=0}^3=(F^0,\dots,F^3) \]
and we also denote by $\eta$ the second quantization of $\eta$,
\[ \eta :=\mathrm{\Gamma}_\bb(\eta):= \op_{n=0}^\infty  \Big(\eta \otimes \dots \otimes  \eta  \Big), \]
whenever no confusion may occur.
Then $ \eta $ is unitary and satisfies $ \eta ^* = \eta , \, \eta ^2 =I $.
We remark that the physical inner product can be written as
\[ \langle\Psi|\Phi\rangle = \langle\Psi,\eta\Phi\rangle. \]
\begin{Def}
For a densely defined linear operator $T$, we denote 
\begin{align}
T^\dagger = \eta T^* \eta.
\end{align}
and call it physical adjoint or $\eta$-adjoint of $T$, where $T^*$ denotes the ordinary adjoint with respect to the positive 
definite inner product $\langle\cdot,\cdot\rangle$. 
\end{Def}
Clearly, it follows that
\begin{align}
\inprod \Psi  {T\Phi } = \inprod{T^\dagger \Psi} \Phi  , \q \Psi \in D(T^\dagger ) , \q \Phi \in D(T) .
\end{align}
The notions on self-adjointness with respect to $\inprod\cdot\cdot$ are defined as follows :
\begin{Def}\label{eta-sa}
\begin{enumerate}[(i)] 
\item A densely defined linear operator $ T $ is symmetric with respect to
$\inprod\cdot\cdot$ or $ \eta $-symmetric if $ T \subset T^\dagger $.

\item A densely defined linear operator $ T $ is self-adjoint with respect to
$\inprod\cdot\cdot$ or $ \eta $-self-adjoint if $ T^\dagger =T $.

\item A densely defined linear operator $ T $ is essentially self-adjoint with respect to
$\inprod\cdot\cdot$ or essentially $ \eta $-self-adjoint if $ \overline{T} $ is $ \eta $-self-adjoint.

\end{enumerate}
\end{Def}

\begin{Lem}\label{eta-lem} 
\begin{enumerate}[(i)] 
\item $ T $ is $ \eta $-symmetric if and only if $ \eta T $ is symmetric.

\item $ T $ is $ \eta $-self-adjoint if and only if $ \eta T $ is self-adjoint. 

\item $ T $ is essentially $ \eta $-self-adjoint if and only if $ \eta T $ is essentially self-adjoint.

\item If $ T $ is $ \eta $-symmetric then $ T $ is closable.

\item Let $ T $ be $ \eta $-self-adjoint and $ \eta T $ is essentially self-adjoint on a subspace $ D $, Then $ D $ is a core of $ T $.
\end{enumerate}
\end{Lem}

\begin{proof} See \cite{MR2533876}.
\end{proof}
Note that the free Hamiltonian $ H_\rm{ph} $ is self-adjoint and $ \eta $-self-adjoint.

The creation operator $ c^\dagger(F) $ with $ F \in \H _\mathrm{ph} $ is a densely defined closed linear operator on $ \F _\mathrm{ph} $ given by
\begin{align}
(c^\dagger(F)\Psi ) ^{(0)} = 0 , \q (c^\dagger (F) \Psi ) ^{(n)} = \sqrt{n} S_n (F \otimes \Psi ^{(n-1)}) , \q n \ge 1 , \q \Psi \in D(c^\dagger(F)), 
\end{align}
where $ S_n $ denotes the symmetrization operator on $ \otimes ^n \H _\rm{ph} $, i.e. $ S_n (\otimes ^n \H _\rm{ph}) = \otimes _\rm{s} ^n \H _\rm{ph} $. We note that $ c^\dagger (F) $ linear in $ F $.
For $f\in L^2(\R)$, we introduce the components of $c^\dagger(\cdot)$ with \textit{lower} indices by
\begin{align}
& c^\dagger_{0} (f) := c^\dagger(f,0,0,0) , \q c^\dagger_{1} (f) := c^\dagger (0,f,0,0) , \\
& c^\dagger_{2} (f) := c^\dagger(0,0,f,0) , \q c^\dagger_{3} (f) := c^\dagger(0,0,0,f).
\end{align}
Then, for $F=(F^\mu)\in\mathcal{H}_\mathrm{ph}$, we have
\begin{align}\label{cov-1}
c^\dagger (F) = c^\dagger_\mu(F^\mu).
\end{align}

Next, we introduce the annihilation operators. 
The annihilation operator is labeled by the element of $\mathcal{H}_\mathrm{ph}^*$, the dual space of
$\mathcal{H}_\mathrm{ph}$. In the ordinary case, the dual space of 
a Hilbert space is naturally identified with the original Hilbert space
via the anti-linear mapping
\begin{align}
F\mapsto \Expect{F,\cdot}.
\end{align}
But, in the present case, the identification mapping we should employ is
\begin{align}
F\mapsto \inprod{F}{\cdot}.
\end{align}
That is, $\phi\in\mathcal{H}_\mathrm{ph}^*$ is identified with the vector 
 $\overline{F}\in\mathcal{H}_\mathrm{ph}$ by the relation
\begin{align}
	\phi(G)=\inprod{\overline{F}}{G}.
\end{align}
Note that there is another vector $F\in\mathcal{H}_\mathrm{ph}$ which is also identified with $\phi\in\mathcal{H}_\mathrm{ph}^*$
via
\begin{align}
	\phi(G)=\Expect{F,G},
\end{align}	
and they are related by
\begin{align}\label{eta-cord}
 \eta F = \overline{F}. 
\end{align}
We employ the notations in which $\mu$-th component of $\overline{F}$ in the direct sum decomposition
\begin{align}
\mathcal{H}_\mathrm{ph} = \op_{\mu=0}^3 L^2(\R^3)
\end{align}
is written as $\overline{F}=(\overline{F}_\mu)$ by \textit{lower} indices and \eqref{eta-cord} can be read as
\begin{align}
\eta_{\mu\nu}F^\nu = \overline{F}_\mu.
\end{align}
Now, the annihilation operator which annihilates one photon with the state $F\in\mathcal{H}_\mathrm{ph}$
is given by $c(\overline{F})=c(\eta F)$, where $c(\cdot)$ is the usual annihilation operator 
of the photon Fock space $\mathcal{F}_\mathrm{ph}$ (for $F\in\mathcal{H}_\mathrm{ph}$, 
$c(F)=c^\dagger(F)^*$). One sees that
\begin{align}
(c^\dagger (F))^\dagger = c (\overline{F}).
\end{align}
It is natural to introduce the components of $c(\cdot)$ with \textit{upper} indices by
\begin{align}
& c^{0} (f) := c(f,0,0,0) , \q c^{1} (f) := c (0,f,0,0) , \\
& c^{2} (f) := c(0,0,f,0) , \q c^{3} (f) := c(0,0,0,f),
\end{align}
for $f\in L^2(\R^3)$. Then, it follows that 
\begin{align}\label{cov-2}
c(\overline{F})=c^\mu(\overline{F}_\mu)=c^\mu(\eta_{\mu\nu}F^\nu).
\end{align}
Form \eqref{cov-1} and \eqref{cov-2}, we have the natural identity
\begin{align}
c^\dagger_\mu(f)=(\eta_{\mu\nu} c^\nu(f))^\dagger,
\end{align}
for all $f\in L^2(\R^3)$.
 
 As is well known, the creation and annihilation operators leave the finite particle subspace 
\begin{align}
\F _{\bb , 0} (\H _\rm{ph}) := \Big\{ \{ \Psi ^{(n)} \} _{n=0} ^\infty \in \F _\rm{ph} \, \Big| \,  \Psi ^{(n)} =0 \, \text{for all $ n\ge n_0 $ with some $n_0$} \Big\} 
\end{align}
invariant and satisfy the canonical commutation relations, which leads in the present case
\begin{align}
[c(\overline{F}) , c^\dagger(F')  ] = \inprod{F}{F'}, \q [c(\overline{F}) , c(\overline{F'} ) ] = [c^\dagger(F)  , c^\dagger(F' ) ]=0,
\end{align}
on $ \F _{\bb, 0} (\H _\rm{ph}) $. 

Let us introduce photon polarization vectors $\{e_\lambda \}_{\lambda=0,1,2,3}$.
Photon polarization vectors are 
$ \R_\kk ^4 $-valued measurable functions defined on $\Real^3$, $ e_\lambda (\cdot)$ ($\lambda=0,1,2,3$),  
such that, for all $ \kk \in M_0  := \R ^3 \backslash \{ (0,0, k^3) \, | \, k^3 \in \R \} $ , 
\begin{align}\label{pol-1}
e _{\lambda} (\kk ) \cdot e _{\sigma} (\kk ) = \eta_{\lambda \sigma} , \q e _{\lambda} (\kk )\cdot k=0 , \q \lambda =1,2,
\end{align}
where the above $\cdot$ means the Minkowski inner product defined by
\[ e_\lambda (\kk)\cdot e _{\sigma} (\kk ) = \eta_{\mu\nu}e^\mu_{\;\;{\lambda}}(\kk ) e^\nu_{\;\;\sigma} (\kk),\q e _{\lambda} (\kk ) \cdot k = \eta_{\mu\nu}e^\mu_{\;\;\lambda} (\kk )  k^\nu,
\q k = (|\kk|, \kk)\in\Real^4,\]
with $e^\mu_{\;\;\lambda}$ being the $\mu$-th component of $e_\lambda$ with respect to the standard basis in $\R^4$.
 Note that such vector valued functions can be chosen so that they are continuous
 on $M_0$, for instance, we may choose
\begin{align}\label{pol-2}
e_0(\kk)=(1,\bvec{0}),\q e_1(\kk)=(0,\ee_1(\kk)),\q e_2(\kk)=(0,\ee_2(\kk)),\q e_3(\kk)=(0,\kk/|\kk|) 
 \end{align}
by using $\{\ee_r (\kk)\}_{r=1,2}$ satisfying the relations
\[\mathbf{e}_r(\kk ) \cdot \mathbf{e}_{r'} (\kk ) = \delta _{rr'} , \q \mathbf{e}_r (\kk ) \cdot \kk =0, \q r,r' =1,2. \]
In this paper, we assume the photon polarization vectors are chosen in this way.

For each $ f\in L^2 (\R _\kk ^3) $ and $ \mu =0,1,2,3 $, we define
\begin{align}
 a^\mu (f) &:= c\big( f e^\mu_{\;\;0} ,\,  f e^\mu_{\;\;1}  ,\, f e^\mu_{\;\;2}  ,\, f e^\mu_{\;\;3}  \big) \no\\
	&= c^\nu(e^\mu_{\;\;\nu} f),
\end{align}
and 
\begin{align}
a_\mu(f):=\eta_{\mu\nu} a^\nu(f).
\end{align}
Then the adjoints are given by
\begin{align}
& a^{\dagger\mu}(f):=(a^\mu(f))^\dagger=\eta^{\sigma\lambda}c^\dagger_\sigma (e^\mu_{\;\:\lambda}f), \\
& a^\dagger_\mu(f):=(a_\mu(f))^\dagger=\eta^{\sigma\lambda}c^\dagger_\sigma (e_{\mu\lambda}f), 
\end{align}
where we have defined 
\begin{align}
e_{\mu\lambda}:=\eta_{\mu\nu}e^{\nu}_{\;\;\lambda}.
\end{align}
The operators $ a_\mu (f) $ and $ a_\mu ^\dagger (f) $ are closed, and satisfy the
Lorentz covariant canonical commutation relations:
\begin{align*}
[a_\mu (f) , a_\nu ^\dagger (g) ] & = \eta_{\mu\nu } \left\langle f,g \right\rangle _{L^2 (\R ^3)} , \\
[a_\mu (f) , a_\nu (g)] & = [a_\mu ^\dagger (f) , a_\nu ^\dagger (g) ] =0 ,
\end{align*}
on $ \F _{\bb , 0} (\H _\rm{ph}) $.

\subsection{Gauge field operator}
For all $ f \in L^2 (\R ^3 _\xx ) $ satisfying $ \hat{f}/\sqrt{\omega} \in L^2 (\R ^3 _\kk ) $, we set 
\begin{align}
A_\mu (f) := 
\frac{1}{\sqrt{2}}\left(a_\mu \left( \frac{\hat{f^*}}{\sqrt{\omega }} \right) + a_\mu ^\dagger \left( \frac{\hat{f}}{\sqrt{\omega}} \right)\right) ,
\end{align}
where $ \hat{f} $ denotes the Fourier transform of $ f $, and $ f^* $ denotes the complex conjugate of $ f $. The functional $ \mathscr{S} (\R ^3 _\xx ) \ni f \mapsto A_\mu (f) $ gives an operator-valued distribution (Cf. \cite{MR2412280} Definition 2.5)
 acting on $ (\F _\rm{ph} , \F _{\bb , 0} (\H _\rm{ph})) $ and it is called the quantized gauge field at time $ t=0 $.
Now, fix $ \chi _\rm{ph} \in L^2 (\R ^3_\xx ) $ which is real and satisfies $\chi_\rm{ph}(\bvec x)=\chi_\rm{ph}(-\xx)$ and $ \hat{\chi _\rm{ph}} / \sqrt{\omega} \in L^2 (\R ^3 _\kk ) $. We set
\begin{align}
& A_\mu (\xx ) := A_\mu (\chi _\rm{ph} ^\xx ) , \\
& \chi _\rm{ph} ^\xx (\yy ) := \chi _\rm{ph} (\yy - \xx ) , \q \yy \in \R ^3 .
\end{align}
$ A_\mu (\xx ) $ is called the point-like quantized gauge field with a momentum cutoff $ \hat{\chi _\mathrm{ph}} $.
Note that $\hat{\chi _\mathrm{ph}}$ is real-valued since $\chi_\rm{ph}(\bvec x)=\chi_\rm{ph}(-\xx)$.
As will be seen later, for real-valued $ f $, the closures of $ A_\mu (f) , \, \mu =0,1,2,3, $ are $ \eta $-self-adjoint but not even normal in the positive definite inner product $\Expect{\cdot,\cdot}$. 
\subsection{Interaction between electrons and gauge field and total Hamiltonian}
Next, we introduce the interaction Hamiltonian and total Hamiltonian 
in the Hilbert space of state vectors 
for the coupled system, which is
taken to be 
\begin{align}
\F _\mathrm{DM}(N) := \H_\rm{D}(N)\ot \F _\mathrm{ph} .
\end{align}
We remark that this Hilbert space can be naturally identified with
\begin{align}
\F _\mathrm{DM} (N)= L^2_\rm{as} (\mathcal{X}^{N} ; \F _\mathrm{ph}) = A_N\int ^{\oplus} _{\mathcal{X}^{N}} d\XX \, \F _\mathrm{ph},
\end{align}
the Hilbert space of $ \F _\mathrm{ph}$-valued functions on $ \mathcal{X}^{N}=\mathcal{X}\times\dots\times\mathcal{X}$
which are square integrable
with respect to the Borel measure (the product measure of the Lebesgue measure on $\R^3$ and the counting measure on $\{1,2,3,4\}$)
 and which are anti-symmetric in the arguments,
that is, the exchange of the $a$-th electron with the $b$-th electron
\begin{align}
(\xx^1,l^1;\dots;\xx^a,l^a;\dots;\xx^b,l^b;\dots;\xx^N,l^N)\mapsto (\xx^1,l^1;\dots;\xx^b,l^b;\dots;\xx^a,l^a\dots;\xx^N,l^N)
\end{align}
gives a minus sign.
The last expression is the constant fibre direct integral with the base space $ (\mathcal{X}^{N} , d\XX ) $ and fibre $\F _\mathrm{ph}$. We freely use this identification. 
 
The mappings $ \XX \mapsto \chi_\rm{ph}^{\xx^a} $ ($a=1,2,\dots,N$) from $ \mathcal{X}^{N} $ to $ \H _\mathrm{ph} $ is strongly continuous, 
and thus we can define a decomposable self-adjoint operator $A_\mu$ by
\begin{align}
A_\mu^a := \int ^\oplus _{\mathcal{X}^{N}} d\XX \, A_\mu(\xx^a ) ,\q \mu=0,1,2,3,\q a=1,2,\dots,N,
\end{align}
acting in $ \int ^\oplus _{\mathcal{X}^{N}} d\XX \, \F _\mathrm{ph}$.
Now we are ready to define the minimal interaction
Hamiltonian $H_1$, between the Dirac particle and the quantized gauge field with the UV cutoff $ \chi_\rm{ph} $. It is given by
\begin{align}
H_1 := q \sum_{a=1}^N \alpha^a \cdot A^a=q\sum_{a=1}^N \alpha^{a\mu} A_\mu^a, \q \alpha^{a\mu}=(1,\bvec{\alpha})=(1,\alpha^{a1},\alpha^{a2},\alpha^{a3}).
\end{align}
The total Hamiltonian of the coupled system is then given by
\begin{align}
H_\mathrm{DM} (V,N) &:= H_0  + H_1, \\
 H_0& :=\overline{H_\mathrm{D} (V,N) + H_\mathrm{ph}}.
\end{align} 
This is the $N$-particle \textit{Dirac-Maxwell Hamiltonian} in the Lorenz gauge. 
We remark that there seems to be no term for the Coulomb interaction between Dirac particles in the Hamiltonian
at a first glance. 
In the Lorenz gauge, this contribution is included in the photon kinetic term $\rm{d\Gamma}_\rm{b}(\omega)$. In fact,
as will be seen later, the Coulomb interaction of Dirac particles emerges from the photon kinetic term via
gauge transformation.

\section{Self-adjointness, Current Conservation and Equations of Motion}\label{time-evolution}
In the present section, we consider the self-adjointness of the Dirac-Maxwell Hamiltonian and the 
time evolution generated by it. The abstract theory developed in Section \ref{abstract}
is applied with $A=N_\bb:=1\otimes\rm{d}\Gamma_\bb(1)$ (See Lemma \ref{C_0-lemma}).
The subspace $V_L$ then becomes
\begin{align}
V_L= R(E_{N_\bb}([0,L]))=\H_\rm{D}(N)\ot \left( \op_{n=0}^L \ot_{s}^n \H_\rm{ph} \right), \q L\in\N,
\end{align}
and $D=\cup_{L}V_L$ is
\begin{align}
D= \hot_{n=0}^\infty \left(\H_\rm{D}(N)\ot \left(\ot_{s}^n \H_\rm{ph} \right)\right).
\end{align}
In what follows, by the word commutator of the linear operator $A$ and $B$, we \textit{always} mean the weak commuatator 
which is denoted by $[A,B]$.
For instance, if we say $[A,B]=C$ on a subspace $\mathcal{D}$, it means that
$\mathcal{D}\subset D(A)\cap D(A^*)\cap D(B)\cap D(B^*)\cap D(C)$ and
\begin{align}
\ip{A^*\psi}{B\phi} - \ip{B^*\psi}{A\phi}= \ip{\psi}{C\phi}, \q \psi,\phi\in \mathcal{D}.
\end{align} 
\subsection{self-adjointness}
\begin{Thm}\label{ess-s.a.}
Under Assumption \ref{ass-V}, $H_{\rm{DM}}(V,N)$ is essentially $\eta$-self-adjoint.
\end{Thm}
To prove Theorem \ref{ess-s.a.}, we prepare some Lemmas.

\begin{Lem}\label{field-lemma}
For all real-valued $f\in L^2(\R)$ with $\hat{f}/\sqrt{\omega}\in L^2(\R)$, the estimate
\begin{align}\label{A_mu-bound}
\norm{A_\mu(f) \Psi} \le 4\sqrt{2} \norm{\frac{\hat{f}}{\sqrt{\omega}}}\norm{(\rm{d}\Gamma_\bb(1)+1)^{1/2}\Psi}
\end{align}
holds for $\Psi\in D(\rm{d}\Gamma(1)^{1/2})$.
\end{Lem}
\begin{proof}
Let $\Psi\in D(\rm{d}\Gamma(1)^{1/2})$. By the well-know inequalities, we have for $\mu=0,1,2,3$,
\begin{align}
\norm{c^\mu(f)\Psi}&\le \norm{f}\norm{\rm{d}\Gamma_\bb(1)^{1/2}\Psi}, \\
\norm{c_\mu^\dagger(f) \Psi} &\le \norm{f}\norm{(\rm{d}\Gamma_\bb(1)+1)^{1/2}\Psi}.
\end{align}
Form these inequalities, we find for $\mu=0,1,2,3$,
\begin{align}\label{annihilation-bound}
\norm{a_\mu\left(\frac{\hat{f}}{\sqrt{\omega}}\right)\Psi}&\le \sum_{\lambda=0}^3\norm{e_{\mu\lambda}\frac{\hat{f}}{\sqrt{\omega}}}\norm{\rm{d}\Gamma_\bb(1)^{1/2}\Psi}\no\\
&\le 4 \norm{\frac{\hat{f}}{\sqrt{\omega}}} \norm{\rm{d}\Gamma_\bb(1)^{1/2}\Psi},
\end{align}
since $|e^\mu_{\;\lambda}(\bvec k)| \le 1$ for almost every $\bvec k\in\R^3$. Similarly, we find
\begin{align}\label{creation-bound}
\norm{a^\dagger_\mu\left(\frac{\hat{f}}{\sqrt{\omega}}\right)\Psi}\le 4 \norm{\frac{\hat{f}}{\sqrt{\omega}}}\norm{(\rm{d}\Gamma_\bb(1)+1)^{1/2}\Psi}
\end{align}
for $\mu=0,1,2,3$. Inequalities \eqref{annihilation-bound} and \eqref{creation-bound} imply \eqref{A_mu-bound}.
\end{proof}

\begin{Lem}\label{C_0-lemma}
The interaction term of the Dirac-Maxwell Hamiltonian $H_1$ satisfies
\begin{enumerate}[(i)]
\item $H_1$ is densely defined, closed.
\item $H_1$ and $H_1^*$ are $N_\bb^{1/2}$-bounded.
\item $\Psi\in V_L:=R(E_{N_\bb}([0,L]))$ implies both $H_1\Psi$ and $H_1^*\Psi$ belong to the subspace $V_{L+1}$.
\end{enumerate}
Therefore, $H_1$ belongs to $\mathcal{C}_0$ class.
\end{Lem}
\begin{proof}
The statement (i) is obvious.

We prove (ii). 
Let $\Psi\in D(N_\bb^{1/2})$. Then, for all $\bvec X\in\mathcal{X}^{N}$,
$\Psi(\bvec X)\in D(\rm{d}\Gamma_\bb(1)^{1/2})$ and
\begin{align}
\int_{\mathcal{X}^{N}}\norm{\rm{d}\Gamma_\bb(1)^{1/2}\Psi(\bvec X)}^2 d\bvec X < \infty.
\end{align}
Thus, one finds $\Psi(\bvec X)\in D(A_\mu(\bvec x^a))$ ($a=1,2,\dots ,N$) and by Lemma \ref{field-lemma} 
\begin{align}
\int_{\mathcal{X}^{N}}\norm{A_\mu(\bvec x^a)\Psi(\bvec X)}^2 d\bvec X&\le 32\norm{\frac{\hat{\chi}_\rm{ph}}{\sqrt{\omega}}}\int_{\mathcal{X}^{N}}
\norm{(\rm{d}\Gamma_\bb(1)+1)^{1/2}\Psi(\bvec X)}^2 d\bvec X \no\\
&= 32\norm{\frac{\hat{\chi}_\rm{ph}}{\sqrt{\omega}}}
\norm{(N_\bb+1)^{1/2}\Psi}^2  \no\\
&<\infty.
\end{align}
This means $\Psi\in D (A^a_\mu)$ ($\mu=0,1,2,3$, $a=1,2,\dots, N$) and 
\begin{align}
\norm{A^a_\mu\Psi} \le 32  \norm{\frac{\hat{\chi}_\rm{ph}}{\sqrt{\omega}}} \norm{(N_\bb+1)^{1/2}\Psi},
\end{align}
and therefore $\Psi\in D(H_1)$ with
\begin{align}
\norm{H_1\Psi}&\le 32|q|\sum_{a=1}^N\norm{\alpha^a{}^\mu} \norm{\frac{\hat{\chi}_\rm{ph}}{\sqrt{\omega}}}\norm{(N_\bb+1)^{1/2}\Psi}\no\\
&\le 32N|q|\norm{\frac{\hat{\chi}_\rm{ph}}{\sqrt{\omega}}}\norm{(N_\bb+1)^{1/2}\Psi}.
\end{align}
Since $H_1^*\Psi=\eta H_1 \eta \Psi$ and $\eta$ commutes with $N_\bb$, $H_1^*\Psi$ satisfies the same estimate.
This proves (ii).

We prove (iii). Let $\Psi\in V_L$ for some $L\ge 0$. It is clear by definition that both $H_1$ and $H_1^*$ create at most one photon.
Thus, $H_1\Psi$ and $H_1^*\Psi$ belong to $V_{L+1}$. 
\end{proof}
\begin{proof}[Proof of Theorem \ref{ess-s.a.}]
One easily sees that $H_0$ is $\eta$-self-adjoint and $H_{\rm{DM}}(V,M)$ is $\eta$-symmetric. 
By utilizing Proposition \ref{s.a.}, Lemma \ref{eta-lem}, and Lemma \ref{C_0-lemma},
we find that it suffices to prove $\eta H_\rm{DM}(V,M)$ has at least one self-adjoint extension.
Following Arai's proof given in Ref. \cite{MR1765584}, we will prove this by applying von Neumann's Theorem (\cite{MR0493420}, Theorem
X.3), which asserts that if there is a conjugation $J$
with $J\eta H_\rm{DM}=\eta H_{\rm{DM}}J$ then $\eta H_\rm{DM}$ has a self-adjoint extension.   

Let $J_\rm{cp}$ be an anti-linear operator in $\mathcal{F}_\rm{DM}(N)$ defined by
\begin{align}
(J_\rm{cp}\Psi)(\bvec x^1,l^1;\dots;\xx^N,l^N):= \sum_{m^1,\dots,m^N} U_{l^1 m^1}\dots U_{l^N m^N}\Psi(-\xx^1,m^1;\dots;-\xx^N,m^N)^*,\label{CP}
\end{align}
where $U$ is a unitary $4\times 4$ matrix satisfying
\eqref{Pauli-matrix1} and \eqref{Pauli-matrix2}. Note that $J_\rm{cp}$ can also be written as
\begin{align}
J_\rm{cp}=\left(\ot_{a=1}^N PUC\right)\otimes \left(\op_{n=0}^\infty \ot_n C \right),
\end{align}
where $P$ is the parity transformation on $L^2(\R^3;\C^4)$ defined by
\[ (Pf)(\bvec x) = f(-\xx) ,\q f\in L^2(\R^3;\C^4) \]
 and $C$ is the complex conjugation.
 Cleary, $J_\rm{cp}$ is anti-linear, norm preserving and satisfies $J_\rm{cp}^2=1$. Thus $J_\rm{cp}$ is a conjugation on $\mathcal{F}_\rm{DM}(N)$. 
 
 We claim 
 \begin{align}\label{claim-1}
 J_\rm{cp} \eta (H_\rm{DM}(V,N)-\sum_{a}\beta^aM^a) \subset \eta( H_\rm{DM}(V,N)-\sum_a\beta^aM^a) J_\rm{cp}.
 \end{align}
 Note that \eqref{claim-1} implies the assertion. In fact, since $J_\rm{cp}$ is a conjugation, \eqref{claim-1} implies
 the operator equality (Ref. \cite{MR1765584}, Lemma 3.1):
 \begin{align}
  J_\rm{cp} \eta (H_\rm{DM}(V,N)-\sum_{a}\beta^aM^a) = \eta( H_\rm{DM}(V,N)-\sum_a\beta^aM^a) J_\rm{cp}.
 \end{align}
 Thus, by von Neumann's Theorem, $\eta(H_\rm{DM}(V,N)-\sum_{a}\beta^aM^a)$ has a self-adjoint extension, say $\eta \tilde{H}$.
 But, since $ \sum_{a}\beta^aM^a $ is a bounded operator, it is obvious that $\eta (\tilde{H}+\sum_{a}\beta^aM^a)$ is a
 self-adjoint extension of $\eta H_\rm{DM}(V,N)$. 
 Note also that $\eta$ and $J_\rm{cp}$ are commuting with each other, and thus \eqref{claim-1} is equivalent to
  \begin{align}\label{claim-2}
 J_\rm{cp} (H_\rm{DM}(V,N)-\sum_{a}\beta^aM^a) \subset ( H_\rm{DM}(V,N)-\sum_a\beta^aM^a) J_\rm{cp}.
 \end{align}
We prove \eqref{claim-2}.
 Take $\Psi\in D(H_0)$. Since the algebraic tensor product
 \begin{align}
 D_0:=\hot_{\rm{as}}^N C^\infty_0(\R^3;\C^4)\hot D(H_\rm{ph})
 \end{align}
 is a core of $H_0$ (Ref. \cite{AraiFock}, Corollary 2-27), there is $\Phi_n\in D_0$
 such that
 \begin{align}
 \Phi_n\to \Psi,\q (H_0-\sum_{a}\beta^aM^a)\Phi_n \to (H_0-\sum_{a}\beta^aM^a)\Psi
 \end{align}
 as $n$ tends to infinity. It is straightforward to verify that $J_\rm{cp}\Phi_n\in D(H_0)$ and 
 \begin{align}\label{app-1}
 J_\rm{cp} (H_0-\sum_{a}\beta^aM^a)\Phi_n  = ( H_0-\sum_a\beta^aM^a) J_\rm{cp}\Phi_n,
 \end{align} 
 if the potential $V$ satisfies Assumption \ref{ass-V} (II). Since $J_\rm{cp}$ is continuous,
 By taking the limit $n\to \infty$ on the both sides of \eqref{app-1}, we find
 $J_\rm{cp}\Psi$ belongs to the domain of $H_0$ and
 \begin{align}
J_\rm{cp} (H_0-\sum_a \beta^a M^a)\Psi =  (H_0-\sum_a \beta^a M^a)J_\rm{cp}\Psi,
 \end{align}
 That is,
\begin{align}\label{H_0-conj}
J_\rm{cp} (H_0-\sum_a \beta^a M^a) \subset  (H_0-\sum_a \beta^a M^a)J_\rm{cp}.
 \end{align}
 Next, take $\Psi \in D(H_1)$. Then, by \eqref{CP}, $J_\rm{cp}\Psi\in D(H_1)$ and
 \begin{align}
 &(J_\rm{cp} H_1 \Psi)(\xx^1,l^1;\dots;\xx^N,l^N)\no\\
 &=\sum_{m^1,\dots,m^N}U_{l^1m^1}\dots U_{l^Nm^N}\sum_{a=1}^N(q\alpha^{a\mu} A_\mu^a\Psi)^*(-\xx^1,m^1;\dots;-\xx^N,m^N) \no\\
 &=q\sum_{a=1}^N\sum_{m^1,\dots,m^N}U_{l^1m^1}\dots (U\alpha^*)_{l^am^a}U_{l^Nm^N}\left\{A_\mu(-\xx^a)\Psi(-\xx^1,m^1;\dots;-\xx^N,m^N)\right\}^* \no\\
 &=q\sum_{a=1}^N\sum_{m^1,\dots,m^N}U_{l^1m^1}\dots (\alpha U)_{l^am^a}U_{l^Nm^N} A_\mu([\chi_\rm{ph}^{-\xx^a}]^*)\Psi(-\xx^1,m^1;\dots;-\xx^N,m^N)^* \no\\
  &=q\sum_{a=1}^N\sum_{m^1,\dots,m^N}U_{l^1m^1}\dots (\alpha U)_{l^am^a}U_{l^Nm^N} A_\mu(\chi_\rm{ph}^{\xx^a})\Psi(-\xx^1,m^1;\dots;-\xx^N,m^N)^* \no\\
    &=q\sum_{a=1}^N\alpha^{a\mu} A_\mu(\xx^a)(J_\rm{cp}\Psi)(\xx^1,m^1;\dots;\xx^N,m^N) \no\\
&=( H_1J_\rm{cp}\Psi)(\xx^1,l^1;\dots;\xx^N,l^N).
 \end{align}
 Hence,
 \begin{align}\label{H_1-conj}
J_\rm{cp} (H_1-\sum_a \beta^a M^a) \subset  (H_1-\sum_a \beta^a M^a)J_\rm{cp}.
 \end{align}
Combining \eqref{H_0-conj} and \eqref{H_1-conj}, we obtain \eqref{claim-2}, which completes the proof.
\end{proof}

\subsection{Time evolution operator  and current conservation}
By Lemma \ref{C_0-lemma}, we can construct the time evolution operator $\{W(t)\}_{t\in\R}$ from the
non-self-adjoint Hamiltonian $H_\rm{DM}(V,N)$. Here we consider the time evolution of the electro-magnetic current density operator 
 $j^\mu(t,f)$ which is an operator-valued distribution defined by
\begin{align}\label{current}
j^\mu(f) := q\sum_{a=1}^N \alpha^{a\mu}\int_{\mathcal{X}^N}^\oplus f(\xx^a) \,d\bvec X, \q f\in\S(\R^3).
\end{align}
As is easily checked, \eqref{current} certainly defines an operator valued tempered distribution (Ref. \cite{MR2412280} Definition 2.5) .
 We remark that $\alpha^{a\mu}$ is the $\mu$-component of the four-component-velocity of the $a$-th electron (See Appendix), and thus, the informal definition of the current density operator
 should be 
 \begin{align}
 j^\mu (\xx)=q\sum_{a=1}^N \alpha^{a\mu}\delta(\xx -\xx^a).
 \end{align}
 If the informal point-like current density $j^{\mu}(\xx)$ is smeared with the smooth function $f$, it gives $j^\mu(f)$ in \eqref{current}.
 
Let $\phi$ be an operator-valued tempered distribution on $\R^3$. We define its derivative in $x^k$ ($k=1,2,3$) by
\begin{align}
\partial_k\phi(f):= \pdfrac{\phi(f)}{x^k} := - \phi\left(\pdfrac{f}{x^k}\right) , \q k=1,2,3.
\end{align} 
If $\phi$ also has a strongly differentiable time-dependence, namely, $\phi$ maps 
$(t,f)\in \R\times \mathscr{S}(\R^3)$ into $\phi(t,f)$, a linear operator in $\F_\rm{DM}$, and   
the mapping $\R\ni t\mapsto \phi(t,f)\Psi\in\F_\rm{DM}(V,N)$, is strongly differentiable in $t\in\R$ for fixed
$f\in\mathscr{S}(\R^3)$ and for $\Psi\in D(\phi(f))$, the strong derivative in $t$ is denoted by
\begin{align}
\partial_0\phi(t,f)\Psi:=\pdfrac{\phi(t,f)}{t}\Psi:=\frac{d}{dt}\phi(t,f)\Psi.
\end{align}   
 
 The most important property of the current density is that it fulfills the 
 conservation equation:
 \begin{Thm}\label{current-conservation-thm}
 The current density $j^\mu(f)$ is in $\mathcal{C}_0$-class for all $f\in \mathscr{S}(\R^3)$ and thus
 the time-dependent current density $j^\mu(t,f):=W(-t)j^\mu(f)W(t)$ exists. The zero-th component 
 $j^0(f)$ is in $\mathcal{C}_1$ and thus satisfies the strong Heisenberg equation of motion. 
 Furthermore, the current density 
 satisfies the conservation equation
 \begin{align}\label{current-conservation}
 \partial_\mu j^{\mu}(t,f):=\pdfrac{j^0(t,f)}{t} + \sum_{k=1,2,3}\pdfrac{j^k(t,f)}{x^k} =0,
 \end{align}
 on $D'$.
 \end{Thm}
 \begin{proof}
 First, we have $j^\mu(f)\in\mathcal{C}_0$, since it is bounded and leave the subspace $V_L$
 invariant.
 To prove $j^0(f)\in\mathcal{C}_1$, we compute the commutator with $iH$ on $D'$:
 \begin{align}\label{commutator-1}
 [iH,j^0(f)]&=\sum_{a}i[\bvec \alpha^a\cdot \pp^a + \beta^aM^a+V^a, j^0(f)] + [i\rm{d}\Gamma_\bb(\omega), j^0(f)]+[iH_1,j^0(f)]\no\\
 &=iq\sum_a\sum_b \alpha^{ak} \left[p^{a}_k, \int_{\mathcal{X}^N}^\oplus f(\xx^b)\,d\XX \right]\no\\
 &=q\sum_a \alpha^{ak} \int_{\mathcal{X}^N}^\oplus  \partial_k f(\xx^a)\,d\XX \no\\
 &=j^{k}(\partial_k f).
 \end{align}
 Since $j^{k}(\partial_k f)\in\mathcal{C}_0$ by the above observation, it follows that $j^0(f)\in\mathcal{C}_1$ with
 $\ad(j^0(f))=j^{k}(\partial_k f)$.
 By Theorem \ref{C_1-equation}, \eqref{commutator-1} also implies that 
 \begin{align}
 \partial_0 j^0(t,f)+\sum_{k=1,2,3}\partial_kj^{k}(t,f)=0
 \end{align}
 on $D'$.
 \end{proof}
 
 \subsection{Time evolution of the gauge field and equation of motion}
The time dependent gauge field $A_\mu(t,f)$ 
 generated by $W(t)$ fulfills the following equation of motion:
 \begin{Thm}
 The gauge field $A_\mu(f)$ is in $\mathcal{C}_2$ class for all $f\in \mathscr{S}(\R^3)$ and the 
 time-dependent field $A_\mu(t,f):=W(-t)A_\mu(f)W(t)$ satisfies the equation of motion
 \begin{align}\label{eq-of-mot}
 \Box A_\mu(t,f):=\partial_\nu\partial^\nu A_\mu(t,f)=j_\mu(t,\chi_\rm{ph}*f), \q t\in\R,\q f\in\mathscr{S}(\R^3).
 \end{align}
 \end{Thm}
 \begin{proof}
 We compute the commutator $\Pi_\mu(f):=[iH,A_\mu(f)]$ on $D'$:
 \begin{align}\label{pi-0}
 \Pi_\mu(f)&:=[iH,A_\mu(f)]=[i\rm{d}\Gamma_\bb(\omega),A_\mu(f)] \no\\
 &=\frac{1}{\sqrt{2}}\left(a_\mu\left((i\omega)\frac{\hat{f^*}}{\sqrt{\omega}}\right)+a^\dagger_\mu\left((i\omega)\frac{\hat{f}}{\sqrt{\omega}}\right)\right).
 \end{align}
 This shows that $\Pi_\mu(f)$ belongs to $\mathcal{C}_0$, and thus $A_\mu(f)\in\mathcal{C}_1$ with $\ad (A_\mu(f))=\Pi_\mu(f)$.
 Again, we have
 \begin{align}
 [iH,\Pi_\mu(f)]&=[i\rm{d}\Gamma_\bb(\omega),\Pi_\mu(f)]+[iH_1,\Pi_\mu(f)]\no\\
 &=\frac{1}{\sqrt{2}}\left(a_\mu\left((i\omega)^2\frac{\hat{f^*}}{\sqrt{\omega}}\right)+a^\dagger_\mu\left((i\omega)^2\frac{\hat{f}}{\sqrt{\omega}}\right)\right)
 +iq\sum_{a}\alpha^{a\nu}\int^\oplus_{\mathcal{X}^N}[A_\nu(\xx^a),\Pi_\mu(f)]\,d\XX \no\\
 &=A_\mu(\Delta f)+\frac{iq}{2}\sum_a\alpha_\mu^a\int_{\mathcal{X}^N}^\oplus \left\{ \ip{\hat{\chi_\rm{ph}^{\xx^a}}}{i\hat{f}}-\ip{i\hat{f^*}}{\hat{\chi_\rm{ph}^{\xx^a}}}\right\} \,d\XX \no \\
 &=A_\mu(\Delta f)-q\sum_a\alpha_\mu^a\int_{\mathcal{X}^N}^\oplus \chi_\rm{ph}*f(\xx^a)\,d\XX \no\\
 &=A_\mu(\Delta f)-j_\mu(\chi_\rm{ph}*f),
\end{align}
which shows that $\Pi_\mu(f)\in\mathcal{C}_1$ with 
\begin{align}\label{pi-1}
 \ad(\Pi_\mu(f))=\ad^2(A_\mu(f))=A_\mu(\Delta f)-j_\mu(\chi_\rm{ph}*f).
 \end{align}
 Therefore, by Theorem \ref{n-th-diff}, we conclude that $A_\mu(t,f)=W(-t)A_\mu(f)W(t)$ is twice differentiable in $t$ on $D'$ and
 \begin{align}
 \frac{d^2}{dt^2}A_\mu(t,f)=A_\mu(t,\Delta f)-j_\mu(t,\chi_\rm{ph}*f),
 \end{align}
 which is equivalent to \eqref{eq-of-mot}.
 \end{proof}
 
 It is clear by Theorem \ref{current-conservation-thm}, \eqref{pi-0} and \eqref{pi-1} that $A_0(t,f)$ ($f\in \mathscr{S}(\R^3)$) is three times differentiable in $t$ and
 \begin{align}
 \partial^\mu\Box A_\mu(t,f)\Psi = 0, \q \Psi\in D'.
 \end{align}
 This implies that the vector-valued function
 \begin{align}
 \partial^\mu A_\mu(t,f)\Psi,\q \Psi\in D'
 \end{align}
 satisfies the Klein-Gordon equation
 \begin{align}
 \Box  \partial^\mu A_\mu(t,f)\Psi = 0,\q \Psi\in D'.
 \end{align}
 It is straightforward to obtain for $f\in\mathscr{S}(\R^3)$
 \begin{align}\label{dA}
 \partial^\mu A_\mu(t,f)\Big|_{t=0}\Psi = -\frac{1}{\sqrt{2}}\left[a_\mu\left(ik^\mu\frac{\hat{f^*}}{\sqrt{\omega}}\right)+
 a^\dagger_\mu\left(ik^\mu\frac{\hat{f}}{\sqrt{\omega}}\right)\right]\Psi,\q \Psi\in D',
 \end{align}
 where $k^0:=\omega(\kk)$.
 
 \section{Gupta-Bleuler's condition and physical subspace}\label{Gupta-Bleuler}
 In the present section, we prove that the operator $\partial_\mu A^\mu(t,f)$ is 
 a generalized free field in the sense of Section \ref{GFF}, and by utilizing this fact, we
 identify the physical subspace following the Gupta-Bleuler's method. Furthermore, 
 the physical subspace naturally induces the Hilbert space with a positive definite metric,
 which we regard as the Hilbert space consisting of all the physical state vectors for 
 the quantum system under consideration.

Let $\mathfrak{h}=D((-\Delta)^{1/4})$ regarded as a Hilbert space by the inner product given by
\begin{align}
\ip{f}{g}_\mathfrak{h}:=\ip{f}{g}_{L^2(\R^3)}+\ip{(-\Delta)^{1/4}f}{(-\Delta)^{1/4}g}_{L^2(\R^3)}.
\end{align} 
A linear operator $T$ in $L^2(\R^3)$ can be 
considered to be an operator $T_\mathfrak{h}$ in $\mathfrak{h}$  with
\begin{align}
D(T_\mathfrak{h})&:=\{ f\in \mathfrak{h}\cap D(T) \,|\, Tf\in\mathfrak{h} \}, \\
T_\mathfrak{h}f&:=Tf,\quad f\in D(T_\mathfrak{h}).
\end{align}
For instance, if we put $T:=(-\Delta)^{1/2}$, then $D(T_\mathfrak{h})=D((-\Delta)^{3/4})$.

In view of \eqref{dA}, it is natural to define a mapping $\Omega:\mathfrak{h}\to\mathcal{C}_0$
 \begin{align}\label{def-omega}
\Omega(f):= -\frac{1}{\sqrt{2}}\left[a_\mu\left(ik^\mu\frac{\hat{f^*}}{\sqrt{\omega}}\right)+
 a^\dagger_\mu\left(ik^\mu\frac{\hat{f}}{\sqrt{\omega}}\right)\right]
 \end{align}
for each $f\in \mathfrak{h}$. 

\subsection{Physical Subspace, Hilbert space and Hamiltonian}
In this subsection, we see that the physical subspace determined by the Gupta-Bleuler method has a positive semi-definite metric and 
naturally induces the positive definite Hilbert space as the completion of the quotient space by
the subspace spanned by the null vectors.
The resulting Hilbert space should be regarded as the physical Hilbert space of the quantum system under consideration.
We see that the original Hamiltonian $H_\rm{DM}(V,N)$ induces the Hamiltonian on the Hilbert space which is essentially self-adjoint. 
\begin{Thm}\label{KG-omega}
Let $T=(-\Delta)^{1/2}$. The mapping $f\mapsto \Omega(f)$ defines an analytic $T_\mathfrak{h}$-free field with
\begin{align}\label{ad-omega}
\ad[\Omega(f)]&= -\frac{1}{\sqrt{2}}\left[a_\mu\left(ik^\mu(i\omega)\frac{\hat{f^*}}{\sqrt{\omega}}\right)+
 a^\dagger_\mu\left(ik^\mu(i\omega)\frac{\hat{f}}{\sqrt{\omega}}\right)\right] \no \\
&\hskip 4mm -\frac{iq}{2}\sum_a \alpha^{a\mu}\int^\oplus_{\mathcal{X}^N}d\bvec X \left[\ip{\frac{\hat{\chi_{\rm{ph}}^{{\bvec x}^a}}}{\sqrt{\omega}}}{\frac{ik_\mu\hat{f}}{\sqrt{\omega}}}
 -\ip{\frac{ik_\mu\hat{f^*}}{\sqrt{\omega}}}{\frac{\hat{\chi_{\rm{ph}}^{{\bvec x}^a}}}{\sqrt{\omega}}}\right]
\end{align}
for $f\in D((-\Delta)^{3/4})$. In particular, 
\begin{align}
\Omega(t,f)=\Omega((\cos t\sqrt{-\Delta})f)+\ad\left[\Omega\left(\frac{\sin t \sqrt{-\Delta}}{\sqrt{-\Delta}}f\right)\right]
\end{align}
on $D'$.
\end{Thm}
\begin{proof}
First, note that $D(T_\mathfrak{h})=D((-\Delta)^{3/4})$ and $D(T_\mathfrak{h}^2)=D((-\Delta)^{5/4})$.
By a direct computation, we obtain
\begin{align}
[iH,[iH,\Omega(f)]]=\Omega(\Delta f)
\end{align}
on $D'$, and thus $\Omega(f)$ belongs to $\mathcal{C}_2$ for $f\in D(T_\mathfrak{h})=D((-\Delta)^{3/4})$ with $\ad^2(\Omega(f))=\Omega(\Delta f)$.
Hence, we see that $\Omega(\cdot)$ is a $T_\mathfrak{h}$-free field by Theorem \ref{n-th-diff}. We prove it is analytic in the sense of Definition \ref{analytic-free}.

For $f\in D(T_\mathfrak{h})$, \eqref{ad-omega} follows from direct computation, which proves (ii) in Definition \ref{analytic-free}.
It is clear by \eqref{def-omega} and \eqref{ad-omega} that the statement (iii) in Definition \ref{analytic-free} is valid.

To prove (i) in Definition \ref{analytic-free}, let $f$ belong to the range of $E_T([1/N,N])$ for some $N\in\Natural$. By Lemma \ref{even-odd},
we have for $n\in\N$
\begin{align}
\ad^{2n}[\Omega(f)]=\Omega((-1)^nT^{2n}f),\label{even-o} \\
\ad^{2n+1}[\Omega(f)]=\ad[\Omega(-1)^n T^{2n}f],\label{odd-o}
\end{align}
on $D'$. By these equations, \eqref{annihilation-bound} and \eqref{creation-bound}, we find that
there is a constant $C>0$ such that
\begin{align}
\norm{\ad^n[\Omega(f)](N_b+1)^{-1/2}}\le C\norm{\frac{ik_\mu(i\omega)^n \hat{f}}{\sqrt{\omega}}}.
\end{align}
But since $\hat{f}$ belongs to the range of $E_\omega([1/N,N])$, one obtains
\begin{align}\label{esti-omega}
\norm{\ad^n[\Omega(f)](N_b+1)^{-1/2}}\le CR^n
\end{align}
for some $R>0$. The equations \eqref{even-o}, \eqref{odd-o}, and \eqref{esti-omega} implies that
$\Omega(f)$ is in $\mathcal{C}_\omega$ class. 

The last assertion immediately follows from Theorem \ref{KG-sol}.
\end{proof}

We have now arrived at the stage to discuss Gupta-Bleuler's definition of physical subspace.
From Theorem \ref{KG-omega} and \eqref{positive-fp}, we can define the positive frequency part of $\Omega(t,f)$ by
\begin{align}\label{omega+}
\Omega^+(t,f)&:=\Omega\left(\frac{e^{-i\sqrt{-\Delta}t }}{2}f\right)-\ad\left[\Omega\left(\frac{e^{-i\sqrt{-\Delta}t }}{2i\sqrt{-\Delta}}f\right)\right]\no\\
&=-\frac{1}{\sqrt{2}}a_\mu\left(\frac{ik^\mu e^{i\omega t}\hat{f^*}}{\sqrt{\omega}}\right)+\frac{iq}{2}\sum_a \int^\oplus_{\mathcal{X}^N} d\bvec X
\ip{\frac{\hat{\chi_\rm{ph}^{{\bvec x}^a}}}{\sqrt{\omega}}}{\frac{e^{-i\omega t}\hat{f}}{\sqrt{\omega}}},
\end{align} 
for $f\in D(T_\mathfrak{h}^{-1})$ on $D'$. (Note that the domain of the operator $T_\mathfrak{h}^{-1}$ is equal to the range of
$T_\mathfrak{h}$ which is $\{(-\Delta)^{1/2}f\,|\, f\in D((-\Delta)^{3/4})\}$). 
Once the positive frequency part of $\Omega(t,f)$ is known, the physical is defined following
Gupta and Bleuler by the relation
\begin{align}
V_\rm{phys}:=\overline{\{\Psi\in D' \,|\, \Omega^+(t,f)\Psi = 0, \;\text{for all}\; t\in\R ,\, f\in D(T_\mathfrak{h}^{-1}) \}}.
\end{align}

To find an explicit expression of the physical subspace $V_\rm{phys}$, we define two unitary operators following Refs. \cite{MR2533876,MR2412280}.
The first one is:
\begin{align}
W:=\op_{n=0}^\infty\ot^n \overline{w}, 
\end{align}
with
\begin{align}
\overline{w}=(w_\mu^{\;\;\nu})=\begin{pmatrix}
1/\sqrt{2} & 0 & 0& -1/\sqrt{2} \\
 0&1&0&0 \\
 0&0&1&0\\
 1/\sqrt{2}&0&0&1/\sqrt{2}
\end{pmatrix} : \mathcal{H}_\rm{ph}^*\to \mathcal{H}_\rm{ph}^*.
\end{align}
To define the other unitary operator we assume that $\hat{\chi_\rm{ph}}\in D(\omega^{-3/2})$. Under this assumption,
we put 
\begin{align}
g^{{\bvec x}^a}=(g^{{\bvec x}^a\mu})=\begin{pmatrix}
 0\\0\\0\\ i\hat{\chi_\rm{ph}^{{\bvec x}^a}}/\omega^{3/2}
\end{pmatrix}\in \mathcal{H}_\rm{ph}
\end{align} and consider the unitary operator $e^{iG}$ with the self-adjoint operator $G$ defined by
\begin{align}\label{G}
G&:=-q\sum_a \int^\oplus_{\mathcal{X}^N} d\bvec X \frac{1}{\sqrt{2}}\overline{\left[c^\mu \left(g^{{\bvec x}^a}_\mu\right)+c^{\dagger}_\mu\left(
g^{{\bvec x}^a\mu}\right) \right]} \no\\
&=-q\sum_a \int^\oplus_{\mathcal{X}^N} d\bvec X \frac{1}{\sqrt{2}}\overline{\left[c^3\left(\frac{i\hat{\chi_\rm{ph}^{{\bvec x}^a}}}{\omega^{3/2}}\right)+c^{\dagger}_3\left(
\frac{i\hat{\chi_\rm{ph}^{{\bvec x}^a}}}{\omega^{3/2}}\right) \right]}.
\end{align}
By these two unitary transformations $W$ and $e^{iG}$, $\Omega^+(t,f)$ is fairly simplified:

\begin{Lem}\label{g-trans}
Let $\hat{\chi_\rm{ph}}\in D(\omega^{3/2})$, $f\in D(T_\mathfrak{h}^{-1})$ and $h\in \mathcal{H}_\rm{ph}$ be
\begin{align}
h:=(h^\mu)=\begin{pmatrix}
i\sqrt{\omega} e^{i\omega t}\hat{f^*} \\ 0 \\ 0 \\ 0
\end{pmatrix}.
\end{align}
Then, we have
\begin{align}
We^{iG}\Omega^+(t,f)e^{-iG}W=c^\mu\left(h_\mu \right)
\end{align}
on $We^{iG}D'$.
\end{Lem}
\begin{proof}
On the subspace $D'$, we have
\begin{align}
\left[iG, \frac{1}{\sqrt{2}}a_\mu\left(\frac{ik^\mu e^{i\omega t}}{\sqrt{\omega}}\hat{f^*}\right)\right]&= \frac{iq}{2}\sum_a \int^\oplus_{\mathcal{X}^N} d\bvec X \left[a_\mu\left(\frac{ik^\mu }{\sqrt{\omega}}\hat{f^*}\right), c^\dagger_\nu(g^{{\bvec x}^a\nu})\right]\no\\
&=\frac{iq}{2}\sum_a \int^\oplus_{\mathcal{X}^N} d\bvec X\ip{\frac{\hat{\chi_\rm{ph}^{{\bvec x}^a}}}{\sqrt{\omega}}}{e^{-it\omega}\hat{f}}.
\end{align}
Thus, it follows that 
\begin{align}
-\frac{1}{\sqrt{2}}e^{-iG}a_\mu\left(\frac{ik^\mu e^{i\omega t}}{\sqrt{\omega}}\hat{f^*}\right)e^{iG}&=-\frac{1}{\sqrt{2}}a_\mu\left(\frac{ik^\mu e^{i\omega t}}{\sqrt{\omega}}\hat{f^*}\right)
+\frac{1}{\sqrt{2}}\left[iG,a_\mu\left(\frac{ik^\mu e^{i\omega t}}{\sqrt{\omega}}\hat{f^*}\right)\right] \no\\
&=-\frac{1}{\sqrt{2}}a_\mu\left(\frac{ik^\mu e^{i\omega t}}{\sqrt{\omega}}\hat{f^*}\right)+\frac{iq}{2}\sum_a \int^\oplus_{\mathcal{X}^N} d\bvec X\ip{\frac{\hat{\chi_\rm{ph}^{{\bvec x}^a}}}{\sqrt{\omega}}}{e^{-it\omega}\hat{f}}\no\\
&=\Omega^+(t,f)
\end{align}
on the subspace $D'$, which implies for $\Psi \in e^{iG}D'$,
\begin{align}\label{part-1}
e^{iG}\Omega^+(t,f)e^{-iG}\Psi=\frac{1}{\sqrt{2}}a_\mu\left(\frac{ik^\mu e^{i\omega t}}{\sqrt{\omega}}\hat{f^*}\right)\Psi.
\end{align}

Introduce four-component vectors $\overline{e}_{\nu}(\bvec k)$ ($\nu=0,1,2,3$ and $\kk\in\Real^3$) by
\begin{align}
\overline{e}^{\mu}_{\;\;\nu}(\kk):= e^\mu_{\;\;\rho}(\kk) w_{\nu}^{\;\;\rho}. 
\end{align}
Then, we have the relations from \eqref{pol-1}, \eqref{pol-2},
\begin{align}
\overline{e}^\mu_{\;\;\nu}(\kk)\overline{e}_{\mu\lambda}(\kk)&=\overline{\eta}_{\nu\lambda}, \no\\
\overline{e}^\mu_{\;\;\nu}(\kk)k_\mu&= \overline{\kappa}_\nu(\kk),
\end{align}
with
\begin{align}
(\overline{\eta}_{\nu\lambda}):=\begin{pmatrix}
0 &0&0& -1 \\
0&1&0&0\\
0&0&1&0\\
-1&0&0&0
\end{pmatrix},\q
(\overline{\kappa}^\mu(\kk))=\begin{pmatrix}
-\sqrt{2}|\kk| \\
0\\
0\\
0
\end{pmatrix}.
\end{align}
Then, one finds
\begin{align}
Wa_\mu\left(\frac{ik^\mu e^{i\omega t}}{\sqrt{\omega}}\hat{f^*}\right)W&=a_\mu\left(w_\rho^{\;\;\mu}\frac{ik^\rho e^{i\omega t}}{\sqrt{\omega}}\hat{f^*}\right)\no\\
&=c^\nu\left(\overline{\kappa}_\nu \frac{ie^{i\omega t}}{\sqrt{\omega}}\hat{f^*}\right),
\end{align}
which, combined with \eqref{part-1} implies
\begin{align}
We^{iG}\Omega^+(t,f)e^{-iG}W\Psi = \frac{1}{\sqrt{2}}c^\mu\left(\overline{\kappa}_\mu \frac{ie^{i\omega t}}{\sqrt{\omega}}\hat{f^*}\right)\Psi=c^\mu(h_\mu)\Psi
\end{align}
for all $\Psi\in We^{iG}D'$.
\end{proof}

Let $\mathcal{F}_{\rm{TL}}$ be the closed subspace of $\F_\rm{DM}(N)$
\begin{align}
\mathcal{F}_{\rm{TL}}:=\H_\rm{D}(N)\otimes\left(\C\Omega\otimes\mathcal{F}_\bb(L^2(\R^3))\otimes\mathcal{F}_\bb(L^2(\R^3))\otimes\mathcal{F}_\bb(L^2(\R^3))\right).
\end{align}
As is well known, for a dense subspace $V\subset L^2(\R^3)$, $\F_\rm{TL}$ is characterized as
\begin{align}\label{vac-character}
\mathcal{F}_{\rm{TL}}=\left.\left\{\Psi\in \bigcap_{f\in V}D(c^0(f)) \,\right|\, c^0(f)\Psi=0, f\in V \right\} = \bigcap_{f\in V}\ker c^0(f).
\end{align}
From Lemma \ref{g-trans}, we can identify the physical subspace $V_\rm{phys}$:
\begin{Thm}
\begin{align}
V_\rm{phys}=e^{-iG}W\mathcal{F}_\rm{TL}.
\end{align}
\end{Thm}
\begin{proof}
It is obvious that the right hand side is included by the left.
To prove the converse, take $\Psi\in D'$ for which $\Omega^+(t,f)\Psi$
vanishes for all $t\in\R$ and $f\in D(T_\mathfrak{h}^{-1})$.
Put $\Phi=We^{iG}\Psi$. Then, 
\begin{align}
c^\mu(h_\mu)\Phi=We^{iG}\Omega^+(t,f)\Psi=0
\end{align}
for all $t\in\R$ and $f\in D(T_\mathfrak{h}^{-1})$. By \eqref{vac-character}, such a $\Phi$ must be an element of $\mathcal{F}_\rm{TL}$. Hence,
$\Psi \in e^{-iG}W\mathcal{F}_\rm{TL}$.
\end{proof}

We investigate the properties of $\Vp$ in detail in order to reveal how the unphysical photon modes are eliminated by
the Gupta-Bleuler formulation. Here, we remark the basic fact 
\begin{align}\label{ip-wick}
&\inprod{c^\dagger_{\mu_1}(F_1^{\mu_1})\dots c^\dagger_{\mu_n}(F_n^{\mu_n})\Omega }{c^\dagger_{\nu_1}(G_1^{\nu_1})\dots c^\dagger_{\nu_m}(G_n^{\nu_m})\Omega}\no\\
&\hskip 4mm =\delta_{nm}\sum_{\sigma\in \mathfrak{S}_n} \inprod{F_1}{G_{\sigma(1)}}\dots\inprod{F_n}{G_{\sigma(n)}} \no\\
&\hskip 4mm=\delta_{nm}\sum_{\sigma\in \mathfrak{S}_n}\eta_{\mu_1\nu_{\sigma(1)}}\ip{F_1^{\mu_1}}{G_{\sigma(1)}^{\nu_{\sigma(1)}}}\dots \eta_{\mu_n\nu_{\sigma(n)}}\ip{F_n^{\mu_n}}{G_{\sigma(n)}^{\nu_{\sigma(n)}}}.
\end{align}
For a closed subspace $\V\subset \H_\rm{ph}$, we define $\mathcal{F}_\bb(\V)$ by
\begin{align}
\mathcal{F}_\bb(\V)&:=\mathcal{\overline{L}}(\{c^\dagger(F_1)\dots c^\dagger(F_n)\Omega, \Omega\,|\, F_1,\dots,F_n\in \V, n\in\N\}), \\
\mathcal{F}^+_\bb(\V)&:=\mathcal{\overline{L}}(\{c^\dagger(F_1)\dots c^\dagger(F_n)\Omega |\, F_1,\dots,F_n\in \V, n\in\N\}),
\end{align}
where $\mathcal{\overline{L}}(\{S\})$ denotes a closed subspace spanned by the vectors in a subset $S\subset\H_\rm{ph}$.
We also define for a direct decomposed subspace $\V=\V_1\op\V_2\subset\H_\rm{ph}$
\begin{align}
\mathcal{F}^+_\bb(\V;\V_1)&:=\mathcal{\overline{L}}(\{c^\dagger(F_1)\dots c^\dagger(F_n)\Omega |\, F_1,\dots,F_n\in \V, n\in\N,\text{ at least one $F_j$ belongs to } \V_1\}) \\
&=\mathcal{F}_\bb(\V)\cap \mathcal{F}_\bb(\V_2)^\perp.
\end{align}
Note that $\F_\bb^+(\V;\V)$ is naturally identified with the tensor product space
\begin{align}
\mathcal{F}^+_\bb(\V_1)\otimes \mathcal{F}_\bb(\V_2).
\end{align}
Let $w$ be a matrix $\eta \overline{w} \eta$ regarded as a bounded operator on $\H_\rm{ph}$:
\begin{align}
w&=(w^\mu_{\;\;\nu})\no\\
&=(\eta^{\mu\nu}w_\nu^{\;\;\rho}\eta_{\rho\nu})=\begin{pmatrix}
-1/\sqrt{2} & 0 & 0& 1/\sqrt{2} \\
 0&1&0&0 \\
 0&0&1&0\\
 1/\sqrt{2}&0&0&1/\sqrt{2}
\end{pmatrix} : \mathcal{H}_\rm{ph}\to \mathcal{H}_\rm{ph}.
\end{align}
Then, $W\F_\rm{TL}$ can be written as
\begin{align}\label{TL}
W\F_\rm{TL}=\H_\rm{D}(N)\ot\F_\bb(w\H_\rm{TL}),
\end{align}
where 
\begin{align}
\H_\rm{TL}:=\{F\in\H_\rm{ph}\,|\, F^0=0\}, 
\end{align}
and
\begin{align}
w\H_\rm{TL}:=\{wF\,|\, F\in\H_{\rm{TL}}\}=\{F\in\H_\rm{ph}\,|\,F^0=F^3\}.
\end{align}
We also define
\begin{align}
\H_\rm{L}&:=\{F\in\H_\rm{TL}\,|\, F^1=F^2=0\}, \\
\H_\rm{T}&:=\{F\in\H_\rm{TL}\,|\, F^3=0\}.
\end{align}
Then, $w\H_\rm{LT}$ is furnished with the direct sum decomposition
\begin{align}\label{direct-sum}
w\H_\rm{TL}=w\H_\rm{L}\op w\H_\rm{T},
\end{align}
with
\begin{align}
w\H_\rm{L}&=\{F\in\H_\rm{ph}\,|\,F^0=F^3, F^1=F^2=0\}, \\
w\H_\rm{T}&=\{F\in\H_\rm{ph}\,|\,F^0=F^3=0\}.
\end{align}
We remark here that if we write the direct sum decomposition of $F,G\in w\H_\rm{TL}$ with respect to \eqref{direct-sum}
by $F=F_T+ F_L$ and $G=G_T+ G_L$, then we find
\begin{align}\label{oth-1}
\inprod{F}{G}=\inprod{F_T}{G_T}=\ip{F_T}{G_T},
\end{align}
since for arbitrary $F_T\in w\H_\rm{T}$ and $G_L\in w\H_\rm{L}$,
\begin{align}\label{oth-2}
\inprod{F_T}{G_L}=0.
\end{align}
\begin{Lem}\label{positive-lemma} The following statements hold:
\begin{enumerate}[(i)]
\item The physical subspace $\Vp$ is non-negative. That is, for all $\Psi\in\Vp$, $\inprod{\Psi}{\Psi}\ge 0$.
\item For $\Psi\in\Vp$, $\inprod{\Psi}{\Psi}=0$ if and only if $\Psi$ belongs to
the closed subspace 
\begin{align}
\mathcal{N} := e^{-iG}(\H_\rm{D}(N)\ot\mathcal{F}^+_\bb(w\H_\rm{TL};w\H_\rm{L})).
\end{align}
\end{enumerate}
\end{Lem}
\begin{proof}
First of all, we remark that $e^{iG}$ is $\eta$-unitary so that 
\begin{align}\label{inprod-inv}
\inprod{e^{-iG}\Psi}{e^{-iG}\Psi}=\inprod{\Psi}{\Psi},
\end{align} 
since it commutes with $\eta$.
Let $\{e_n\}_{n\in\N}$ and $\{f_m\}_{m\in\N}$ be complete orthonormal systems on $w\H_\rm{L}$ and 
$w\H_\rm{T}$, respectively. 
Take arbitrary $\Psi\in \F_\bb(w\H_\rm{TL})$. Then, $\Psi$ can be written as
\begin{align}
\Psi = \sum_{j=0}^\infty\sum_{k=0}^\infty\sum_{n_1,\dots,n_j\in\N}\sum_{m_1,\dots,m_k\in\N}\alpha_{j,k}(n_1,\dots,n_j ; m_1,\dots, m_k) c^\dagger(e_{n_1})\dots c^\dagger(e_{n_j})
c^\dagger(f_{m_1})\dots c^\dagger(f_{m_k})\Omega,
\end{align}
with coefficient mappings $\alpha_{j,k}$ 
\begin{align}
\alpha_{j,k}: \N^j\times \N^k \to \C,\q j,k \ge 0
\end{align}
satisfying 
\begin{align}
\ip{\Psi}{\Psi}=\sum_{j,k=0}^\infty\sum_{n_1,\dots,n_j}\sum_{m_1,\dots,m_k}C_{j,k}(n_1,\dots,n_j;m_1,\dots,m_k)|\alpha_{j,k}(n_1,\dots,n_j ; m_1,\dots, m_k)|^2 < \infty.
\end{align}
Here, 
\begin{align}
C_{j,k}:\N^j\times \N^k \to \C,\q j,k \ge 0
\end{align}
 is defined as follows. Let $\{n_1,\dots,n_j\}=\{\nu_1,\dots,\nu_p\}$ and $\{m_1,\dots,m_k\}=\{\mu_1,\dots,\mu_q\}$
with $\nu$'s and $\mu$'s are mutually different. We denote by $N_l$ ($l=1,2,\dots,p$) the number of $\nu_l$'s in $\{n_1,\dots, n_p\}$
and denote by $M_{l'}$ ($l'=1,2,\dots,q$) the number of $\mu_{l'}$'s in $\{m_1,\dots, m_p\}$, that is,
\begin{align}
N_l&:=\sharp\{ i\,|\, \nu_l=n_i \} \\
M_{l'}&:=\sharp\{ i\,|\, \mu_{l'}=m_i \}.
\end{align}
Then, 
\begin{align}
C_{j,k}(n_1,\dots,n_j;m_1,\dots,m_k):=N_1!\dots N_p! M_1!\dots M_q!.
\end{align}
By \eqref{ip-wick} and \eqref{oth-2} one obtains
\begin{align}\label{inprod}
\inprod{\Psi}{\Psi}=\sum_{k=0}^\infty\sum_{m_1,\dots,m_k}C_{0,k}(m_1,\dots,m_k)|\alpha_{0,k}( m_1,\dots, m_k)|^2.
\end{align}

From \eqref{inprod-inv} and \eqref{inprod}, the assertion (i) is obvious. 
The assertion (ii) also immediately follows from the fact that $\Psi\in \F_\bb^+(w \H_\rm{TL}; w\H_\rm{L})$ if and only if
$\alpha_{0,k}=0$ for all $k\ge 0$.
 \end{proof}
 
By Lemma \ref{positive-lemma}, the quotient vector space $\Vp/\mathcal{N}$ becomes a pre-Hilbert space with 
respect to the naturally induced metric from $\eta$, and its completion
 \begin{align}
 \H_\rm{phys}:=\overline{\Vp/\mathcal{N}}
 \end{align}
is a Hilbert space with the induced metric, which we also denote by $\inprod{\cdot}{\cdot}$. This Hilbert space $\H_\rm{phys}$,
which we call physical Hilbert space, consists of all 
the physical state vectors $\Psi$ satisfying $\inprod{\Psi}{\Psi}\ge 0$.  
Let $\V$ be an orthogonal complement
to $\mathcal{N}$ in $\V_\rm{phys}$:
\begin{align}
\V_\rm{phys}=\V\oplus \mathcal{N},
\end{align}
and $P$ be an orthgonal projection onto the subspace $\V$.
Then, via unitary transformation $U:\V\to \H_\rm{phys}$,
\begin{align}\label{U}
U:\Psi \mapsto [\Psi],\q U^{-1}:[\Psi]\mapsto P\Psi, 
\end{align}
$(\V,\ip{\cdot}{\cdot})$ is identified with $(\H_\rm{phys},\inprod{\cdot}{\cdot})$ ($[\Psi]$ denotes the equivalent class to which $\Psi$ belongs.). 
We remark that $\V$ is equal to $e^{-iG}\H_\rm{D}(N)\ot\F_\bb(wH_\rm{T})$ by Lemma \ref{positive-lemma} and
that $\eta$ is identity on $\V$, namely, $\inprod{\Psi}{\Phi}=\ip{\Psi}{\Phi}$ for all $\Psi,\Phi\in\V$.

For a bounded operator $A$ in $\F_\rm{DM}$ satisfying $AV_\rm{phys}\subset V_\rm{phys}$
and $A\mathcal{N}\subset\mathcal{N}$, the bounded operator $\tilde{A}$ in $\H_\rm{phys}$ is 
defined by the relation
\begin{align}
\tilde{A}[\Psi]:=[A\Psi],\q\Psi \in V_\rm{phys}.
\end{align}
In this case, the operator $\tilde{A}$ is identified with $PAP$ by the unitary transformation $U$ given by \eqref{U}.
On the other hand, if $A$ is unbounded, the situation may become complicated since we have to be careful to deal with the operator domain. 
It may happen, in general, that even if $D(A)$ is dense in $\F_\rm{DM}$, $PD(A)=\{0\}$. 
Thus, it is not a trivial problem to define the physical Hamiltonian in a suitable manner.
We define the physical Hamiltonian by the generator of the time evolution which is naturally induced by $W(t)$.

\begin{Lem}\label{group-prop}
\begin{enumerate}[(i)]
Let $\psi\in D$ and $s,t,\lambda\in\R$, and $B\in\mathcal{C}_0$. Then,
\item $e^{i\lambda G}\psi\in D(W(s)W(t))$ and
\begin{align}
W(s)W(t)e^{i\lambda G}\psi=W(s+t)e^{i\lambda G}\psi.
\end{align}
\item $e^{i\lambda G}\psi\in D(W(-t)B(s)W(t))$ and
\begin{align}
W(-t)B(s)W(t)e^{i\lambda G}\psi = B(s+t)e^{i\lambda G}\psi.
\end{align}
\end{enumerate}
\end{Lem}
\begin{proof}
\begin{enumerate}[(i)]
\item First, the estimate for some $C>0,b>0$ and $L_\psi > 0$
\begin{align}
&\norm{\sum_{l,m,n}e^{-isH_0}U_l(s,0)e^{-itH_0}U_m(t,0)\frac{(i\lambda)^n}{n!}G^n\psi }\no\\
\le &\sum_{l,n,m}\frac{|\lambda|^n }{n!} \norm{U_l(s+t,t)U_m(t,0)G^n\psi} \no\\
\le &\sum_{l,m,n}\frac{|\lambda|^n }{n!} \frac{|s|^m}{m!}\frac{|t|^n}{n!}C^{m+n} (L_\psi+b(l+n+m-3)+1)^{1/2}\cdots(L_\psi+1)^{1/2}\norm{\psi}\no\\
\le &\sum_{N=0}^\infty \frac{1}{N!}(|\lambda|+C|s|+C|t|)^N(L_\psi+b(N-3)+1)^{1/2}\cdots(L_\psi+1)^{1/2}\norm{\psi},
\end{align} 
shows that the series 
\begin{align}
\sum_{l,m,n}e^{-isH_0}U_l(s,0)e^{-itH_0}U_m(t,0)\frac{(i\lambda)^n}{n!}G^n\psi
\end{align}
is absolutely convergent and thus defines a vector in $\mathcal{F}_\rm{DM}(V,N)$.
On the other hand, one similarly sees that the series
\begin{align}
&\sum_{n}\frac{(i\lambda)^n }{n!} G^n\psi, \\
&\sum_{m,n}e^{-itH_0}U_m(t,0)\frac{(i\lambda)^n}{n!}G^n\psi 
\end{align}
are also absolutely convergent. Since the operators $e^{i\lambda G}$, $W(s)$ and $W(t)$
are closed operators, this implies that $\psi\in D(W(s)W(t)e^{i\lambda G})$ and 
\begin{align}\label{expansion-of-LHS}
W(s)W(t)e^{i\lambda G}=\sum_{l,m,n}e^{-isH_0}U_l(s,0)e^{-itH_0}U_m(t,0)\frac{(i\lambda)^n}{n!}G^n\psi.
\end{align}
But the right hand side of \eqref{expansion-of-LHS} is equal to
\begin{align}
\sum_{N,n}e^{-i(s+t)H_0}U_N(s+t,0)\frac{(i\lambda)^n}{n!}G^n\psi = W(s+t)e^{i\lambda G\psi},
\end{align}
since the equality
\begin{align}\label{sum-rule}
\sum_{l+m=N}U_l(s+t,t)U_m(t,0)=U_N(s+t,0)
\end{align}
hols on $D$.
\item The proof is so similar to that of (i) that we omit it.
\end{enumerate}
\end{proof}

\begin{Lem}
For all $\psi \in e^{i\lambda G}D$ ($\lambda\in\R$) and $t\in\R$, 
\begin{align}\label{unitarity}
\inprod{W(t)\psi}{W(t)\psi}=\inprod{\psi}{\psi}.
\end{align}
\end{Lem}
\begin{proof}
By the easily checked equality
\begin{align}
U_n(s,t)^\dagger=U_n(t,s),\q n=0,1,2,\dots,
\end{align}
on $D$, and \eqref{sum-rule}, we have for $\psi\in e^{i\lambda G} D$ with $\phi\in D$,
\begin{align}
\inprod{W(t)\psi}{W(t)\psi}&=\inprod{W(t)e^{i\lambda G}\phi}{W(t)e^{i\lambda G}\phi}\no\\
&=\sum_{k,l,n,m}\frac{(i\lambda )^k}{k!}\frac{(i\lambda )^l}{l!}\inprod{G^k\phi}{U_n(0,t)U_m(t,0)G^l\phi}\no\\
&=\sum_{k,l}\sum_{N=0}^\infty\frac{(i\lambda )^k}{k!}\frac{(i\lambda )^l}{l!}\inprod{G^k\phi}{U_N(0,0)G^l\phi}\no\\
&=\inprod{e^{i\lambda G}\phi}{U(0,0)e^{i\lambda G}\phi}\no\\
&=\inprod{\psi}{\psi}.
\end{align}
This completes the proof.
\end{proof}
We employ the notations
\begin{align}
V_{\rm{phys},0}&:=e^{-iG}\left(W\F_\rm{TL}\cap D\right)=e^{-iG}\left(\H_\rm{D}(N)\ot \F_\bb(w\H_\rm{TL})\cap D\right), \\
\mathcal{N}_0&:=e^{-iG}\left(\H_\rm{D}(N)\ot\F_{\bb}(wH_\rm{T})\cap D\right).
\end{align}
Clearly, $V_\rm{phys,0}$ and $\mathcal{N}_0$ are dense subspaces of
$V_\rm{phys}$ and $\mathcal{N}$ respectively.
Then, we have

\begin{Lem}\label{well-def}
\begin{enumerate}[(i)]
\item $W(t)V_{\rm{phys},0}\subset V_\rm{phys}$.
\item $W(t)\mathcal{N}_0\subset \mathcal{N}$.
\end{enumerate}
\end{Lem}
\begin{proof}
Let $\Psi\in\F_{\bb,0}(w\H_\rm{TL})$. Then, the assertion (i) follows from the computation using Lemma \ref{group-prop}
\begin{align}
\Omega^+(s,f)W(t)e^{-iG}\Psi=W(t)\Omega^+(s+t,f)e^{-iG}\Psi = 0,
\end{align}
for $e^{-iG}\Psi\in V_\rm{phys}$. This means $W(t)e^{-iG}\Psi\in V_\rm{phys}$.
The assertion (ii) is obvious from \eqref{unitarity}.
\end{proof}

From Lemma \ref{well-def}, we can define the densely defined operator $\{\tilde{W}(t)\}_{t\in\R}$ on $\H_\rm{phys}$ as follows.
Let $[V_{\rm{phys},1}]$ be a dense subspace of $\H_\rm{phys}$ spanned by the vectors equivalent to a vector belonging to the subspace
\begin{align}
V_{\rm{phys},1}:=\{W(t)\Psi\,|\, t\in\R, \, \Psi\in V_{\rm{phys},0} \}.
\end{align}
On the subspace $[V_{\rm{phys},1}]$, we define a linear operator $\tilde{W}(s)$ ($s\in\R$) by  
\begin{align}
D(\tilde{W}(s))&:=[V_{\rm{phys},1}]=\{[\Psi]\,|\,\Psi=W(t)\Phi \text{ for some $t\in\R$ and $\Psi\in V_{\rm{phys},0}$}\} ,\\
\tilde{W}(s)[W(t)\Phi]&:=[W(s+t)\Phi].
\end{align}

The next lemma is crucial:
\begin{Lem}\label{unitary-group}
$\tilde{W}(t)$ defines a strongly continuous one-parameter unitary group $\{U(t)\}_t$ on $\H_\rm{phys}$.
\end{Lem}
\begin{proof}
First, we see $\tilde{W}(t)$ has the unitary extension. In fact, for all $\Psi\in V_{\rm{phys},1}$, it follows that
\begin{align}
\inprod{\tilde{W}(t)[\Psi]}{\tilde{W}(t)[\Psi]}=\inprod{W(t)\Psi}{W(t)\Psi}=\inprod{\Psi}{\Psi}=\inprod{[\Psi]}{[\Psi]}
\end{align}
which shows $\tilde{W}(t)$ is isometric. Furthermore, since $R(\tilde{W}(t))\supset [V_{\rm{phys},1}]$, $\tilde{W}(t)$ has
the unitary extension which we denote by $U(t)$.

To prove that $\{U(t)\}_t$ has the stated property, let $s,t\in\R$. Then,
\begin{align}
\tilde{W}(s)\tilde{W}(t)[\Psi]=\tilde{W}(s+t)[\Psi],\q \Psi\in V_{\rm{phys},1}.
\end{align}
Hence, we have the group property
\begin{align}
U(s)U(t)=U(s+t).
\end{align}
The continuity follows from the fact that $\tilde{W}(\cdot)[\Psi]$ is strongly continuous for $\Psi\in V_{\rm{phys},1}$.
\end{proof}

Lemma \ref{unitary-group} implies that the physical Hamiltonian $H_\rm{phys}$ should be defined as:
\begin{Def}
The physical Hamiltonian $H_\rm{phys}$ is the generator of $\{U(t)\}_t$.
\end{Def}

The following Lemma follows from lengthy but straightforward computation:
\begin{Lem}\label{transf}
On the subspace 
\begin{align}
\mathcal{D}_1:=\left(\hot_{\rm{as}}^N C^\infty_0(\R^3;\C^4)\right)\hat{\otimes}\F_{\bb,\rm{fin}}(C^\infty_0(\R^3;\C^4))
\end{align}
 one finds
\begin{align}\label{H-transf}
e^{iG}H_\rm{DM}(V,N)e^{-iG}&=H_0+q\sum_{a}\alpha^{a\mu}\int^\oplus_{\mathcal{X}^N}d\bvec X\frac{1}{\sqrt{2}}
\left[ c^\nu\left(Q_{\mu\nu}(\xx^a)\right)+c^\dagger_\nu\left(Q_\mu^{\;\;\nu}(\xx^a)\right)\right] \no\\
&\hskip 4mm -\frac{1}{2}q^2 \sum_{a,b}\alpha^{a\mu}\int^\oplus_{\mathcal{X}^N}d\bvec X \ip{k_\mu g^{\xx^a}_\nu}{g^{\xx^b \nu}},
\end{align}
where
\begin{align}
Q_\mu^{\;\;\nu}(\xx):=ik_\mu g^{\xx \nu}+e_{\mu}^{\;\;\nu}\frac{\chi_\rm{ph}^{\xx}}{\sqrt{\omega}}\in\mathcal{H}_\rm{ph},\q \xx\in\Real^3.
\end{align}
\end{Lem}
\begin{proof}
Let us compute the commutators on $\mathcal{D}_1$:
\[[iG,H_\rm{D}(V,N)],\q [iG,\rm{d}\Gamma_\bb (\omega)],\q [iG, H_1].\] 
The first one is 
\begin{align}
[iG,H_\rm{D}(V,N)]&=\left[iG,\sum_a(\bvec \alpha^a\cdot\bvec p^a+\beta^a M^a + V^a)\right] \no\\
&=\left[iG,\sum_a \bvec \alpha^a\cdot\bvec p^a\right] \no\\
&=-\frac{iq}{\sqrt{2}}\sum_{a,b}\left[\int_{\mathcal{X}^N}^\oplus d\bvec X\,\left\{c^\mu(g^{{\bvec x}^a}_\mu)+c^\dagger_\mu(g^{{\bvec x}^a\mu})\right\},\bvec \alpha^b\cdot \bvec p^b\right] \no\\
&=\frac{q}{\sqrt{2}}\sum_{a}\alpha^{aj}\int_{\mathcal{X}^N}^\oplus\left\{c^\mu(ik^j g^{{\bvec x}^a}_\mu)+c^\dagger_\mu(ik^j g^{{\bvec x}^a\mu})\right\}.
\end{align}
The next one is
\begin{align}
[iG,\rm{d}\Gamma_\bb (\omega)]&=-\frac{iq}{\sqrt{2}}\sum_a \int^\oplus_{\mathcal{X}^N} d\bvec X\,\left[c^\mu(g^{{\bvec x}^a}_\mu)+c^\dagger_\mu(g^{{\bvec x}^a\mu}),\rm{d}\Gamma_\bb (\omega)\right]\no\\
&=-\frac{q}{\sqrt{2}}\sum_a \int^\oplus_{\mathcal{X}^N} d\bvec X\,\left\{c^\mu(i\omega g^{{\bvec x}^a}_\mu)+c^\dagger_\mu(i\omega g^{{\bvec x}^a\mu})\right\}.
\end{align}
The last one is
\begin{align}
[iG,H_1]&=-\frac{iq}{\sqrt{2}}\sum_{a} \int^\oplus_{\mathcal{X}^N} d\bvec X\,\left[c^\mu(g^{{\bvec x}^a}_\mu)+c^\dagger_\mu(g^{{\bvec x}^a\mu}),q\sum_b\alpha^{b\mu}A_\mu({\bvec x}^b)\right]\no\\
&=-\frac{iq^2}{2}\sum_{a,b}\alpha^{b\nu} \int^\oplus_{\mathcal{X}^N} d\bvec X\,\left[c^\mu(g^{{\bvec x}^a}_\mu)+c^\dagger_\mu(g^{{\bvec x}^a\mu}),c^\lambda(e_{\nu\lambda}\chi_\rm{ph}^{{\bvec x}^b}\omega^{-1/2})+c^\dagger_\lambda(e_{\nu}^{\;\;\lambda}\chi_\rm{ph}^{{\bvec x}^b}\omega^{-1/2})\right]\no\\
&=-\frac{iq^2}{2}\sum_{a,b}\alpha^{b\nu} \int^\oplus_{\mathcal{X}^N} d\bvec X\,\left\{ \eta^{\mu}_{\;\;\lambda}\ip{g_\mu^{{\bvec x}^a}}{e_{\nu}^{\;\;\lambda}\chi_\rm{ph}^{{\bvec x}^b}\omega^{-1/2}}-\eta^\lambda_{\;\;\mu} \ip{e_{\nu\lambda}\chi_\rm{ph}^{{\bvec x}^b}\omega^{-1/2}}{g^{{\bvec x}^a\mu}} \right\}\no\\
&=q^2\sum_{a,b}\alpha^{b\nu} \int^\oplus_{\mathcal{X}^N} d\bvec X\,\rm{Im}\, \ip{g_\mu^{{\bvec x}^a}}{e_{\nu}^{\;\;\mu}\chi_\rm{ph}^{{\bvec x}^b}\omega^{-1/2}}\no\\
&=0.
\end{align}

Taking the commutator with $iG$ once again, we find
\begin{align}\label{second-comm-hd}
[iG,[iG,H_\rm{D}(V,N)]]&=-\frac{iq^2}{2}\sum_{a,b}\alpha^{bj} \int^\oplus_{\mathcal{X}^N} d\bvec X\,\left[c^\nu(g^{{\bvec x}^a}_\nu)+c^\dagger_\nu(g^{{\bvec x}^a\nu}),c^\mu(ik^j g^{{\bvec x}^b}_\mu)+c^\dagger_\mu(ik^j g^{{\bvec x}^b\mu})\right]\no\\
&=-\frac{iq^2}{2}\sum_{a,b}\alpha^{bj} \int^\oplus_{\mathcal{X}^N} d\bvec X\,\left\{\ip{g^{{\bvec x}^a}_\mu}{ik^j g^{{\bvec x}^b\mu}}-\ip{ik^j g^{{\bvec x}^b}_\mu}{g^{{\bvec x}^a\mu}}\right\}\no\\
&=q^2\sum_{a,b}\alpha^{bj} \int^\oplus_{\mathcal{X}^N} d\bvec X\,\rm{Im}\,\ip{g^{{\bvec x}^a}_\mu}{ik^j g^{{\bvec x}^b\mu}},
\end{align}
and
\begin{align}\label{second-comm-gamma}
[iG,[iG,\rm{d}\Gamma_\bb (\omega)]]&=\frac{iq^2}{2}\sum_{a,b} \int^\oplus_{\mathcal{X}^N} d\bvec X\, \left[c^\nu(g^{{\bvec x}^a}_\nu)+c^\dagger_\nu(g^{{\bvec x}^a\nu}),c^\mu(i\omega g^{{\bvec x}^b}_\mu)+c^\dagger_\mu(i\omega g^{{\bvec x}^b\mu})\right]\no\\
&=\frac{iq^2}{2}\sum_{a,b} \int^\oplus_{\mathcal{X}^N} d\bvec X\, \left\{\ip{g^{{\bvec x}^a}_\nu}{i\omega g^{{\bvec x}^b\mu}}-\ip{i\omega g^{{\bvec x}^b}_\mu}{g^{{\bvec x}^a\nu}}\right\} \no\\
&=-q^2\sum_{a,b} \int^\oplus_{\mathcal{X}^N} d\bvec X\,\rm{Im}\,\ip{g^{{\bvec x}^a}_\mu}{i\omega g^{{\bvec x}^b\mu}}.
\end{align}
Note that the final expressions \eqref{second-comm-hd} and \eqref{second-comm-gamma} are commuting with $iG$.

Summarizing these results in a covariant manner, one has on the subspace $\mathcal{D}_1$
\begin{align}
[iG,H_0]&=\frac{q}{\sqrt{2}}\sum_a \alpha^{a\nu}\int^\oplus_{\mathcal{X}^N} d\bvec X\, \left\{c^\mu(ik_\nu g^{{\bvec x}^a}_\mu)+c^\dagger_\mu(ik_\nu g^{{\bvec x}^a\mu})\right\}, \no\\
[iG,H_1]&=0, \no\\
[iG,[iG,H_0]]&=q^2\sum_{a,b}\alpha^{b\nu} \int^\oplus_{\mathcal{X}^N} d\bvec X\,\rm{Im}\,\ip{g^{{\bvec x}^a}_\mu}{ik_\nu g^{{\bvec x}^b\mu}},
\end{align}
and higher commutators with $iG$ vanish.

From the formula
\begin{align}
e^{i\lambda G} H_\rm{DM}(V,N)e^{-i\lambda}G = \sum_{n=0}^\infty \frac{\lambda^n}{n!}[iG ,[ iG,[\dots [iG, H_\rm{DM}(V,N)]\dots] ],
\end{align}
we have \eqref{transf} at least heuristically. To make this computation a mathematical proof, we have to be careful on the operator domains.
But this can be done by the argument of closedness, which is letft to the reader.
\end{proof}

Regard the above $Q_\mu^{\;\;\nu}(\xx)$ as a $\nu$-th component of four-component vector $Q_\mu(\xx)=(Q_\mu^{\;\;\nu}(\xx))\in\H_\rm{ph}$ ($\mu=0,1,2,3$).
Then, it immediately follows that
\begin{Lem}\label{inv}
The four component vectors $Q_\mu(\xx)$ ($\mu=0,1,2,3$) belong to $w\H_\rm{TL}$. In particular, $H_\rm{DM}(V,N)|_{e^{-iG}\mathcal{D}_1}$ leaves $V_\rm{phys}$ and $\mathcal{N}$ invariant.
\end{Lem} 
\begin{proof}
By a direct computation, we find 
\begin{align}\label{Q-concrete}
Q_0(\xx) = (Q_0^{\;\;\nu}(\xx))_\nu=\begin{pmatrix}
\dfrac{\hat{\chi^\xx}}{\sqrt{\omega}} \\ 0\\0\\ \dfrac{\hat{\chi^\xx}}{\sqrt{\omega}}
\end{pmatrix}\in wH_\rm{L}, \q
Q_k(\xx)=(Q_k^{\;\;\nu}(\xx))_\nu=\begin{pmatrix}
0 \\ e_k^{\;\;1} \dfrac{\hat{\chi^\xx}}{\sqrt{\omega}}\\ e_k^{\;\;2} \dfrac{\hat{\chi^\xx}}{\sqrt{\omega}}\\ 0
\end{pmatrix}\in w H_\rm{T} ,\q k=1,2,3.
\end{align}
Note that the second term in \eqref{H-transf} can be written as
\begin{align}
q\sum_{a}\int^\oplus_{\mathcal{X}^N}d\bvec X\frac{1}{\sqrt{2}}
\left[ c\left(\overline{Q}_{0}(\xx^a)\right)+c^\dagger\left(Q_0(\xx^a)\right)\right] 
+q\sum_k\sum_{a}\alpha^{ak}\int^\oplus_{\mathcal{X}^N}d\bvec X\frac{1}{\sqrt{2}}
\left[ c\left(\overline{Q}_k(\xx^a)\right)+c^\dagger\left(Q_k(\xx^a)\right)\right] 
\end{align}
where
\begin{align}
\overline{Q}_\mu(\xx) = (Q_{\mu\nu}(\xx))_\nu = (\eta_{\nu\lambda} Q_\mu^{\;\;\lambda}(\xx) )_\nu,
\end{align}
and $F\in wH_\rm{L}$ implies $\overline{F}\in wH_\rm{S}$ while $wH_\rm{T}$ is $\eta$-invariant.
 Then, the assertions obviously follow.
\end{proof}

\begin{Thm}\label{ess-sa}
The operator $\tilde{H}$ defined by the relations
\begin{align}
D(\tilde{H})&:=[\left(e^{-iG}\mathcal{D}_1\right)\cap V_\rm{phys}], \\
\tilde{H}[\Psi] &:= [H_\rm{DM}(V,N)\Psi] ,
\end{align}
is essentially self-adjoint and the unique self-adjoint extension is 
equal to $H_\rm{phys}$.
\end{Thm}
\begin{proof}
 For $\Psi\in (e^{-iG}\mathcal{D}_1)\cap V_\rm{phys}$, 
 \begin{align}
 \frac{d}{dt}U(t)[\Psi]\Bigg|_{t=0}=[-iH_\rm{DM}(V,N)\Psi]=-i\tilde{H}[\Psi].
 \end{align}
 Thus the assertion follows from Proposition \ref{a-or-b}.
\end{proof}
\subsection{Relation to the Coulomb gauge Hamiltonian}
In the Coulomb gauge, the Hilbert space $\mathcal{H}_\rm{Coulomb}$ of state vectors of the quantum system considered is naturally taken to be 
\begin{align}
\H_\rm{Coulomb}:=\wedge_{a=1}^NL^2(\R^3;C^4)\ot \F_\bb(L^2(\R^3;\C^2)).
\end{align}
In the Coulomb gauge, only two of the the photon polarization, that is, the two transverse ones, are
in the model, so that the photon Hilbert space is chosen to be the boson Fock space over $L^2(\R^3;\C^2))$,
the target space $\C^2$ representing the two photon polarization degrees of freedom.
In our present formulation in the Lorentz gauge, the closed subspace 
\begin{align}
\H_\rm{D}(N)\ot \F_\bb(\H_\rm{T})
\end{align} 
is naturally identified with $\H_\rm{Coulomb}$. But $w$ is just an identity matrix when restricted on $\H_\rm{T}$,
it is equal to the closed subspace of $\F_\rm{DM}(V,N)$
\begin{align}
\H_\rm{D}(N)\ot \F_\bb(w\H_\rm{T}).
\end{align}
As we have ever seen, the physical Hilbert space with positive definite metric is naturally identified
with the closed subspace
\begin{align}
\mathcal{V}=e^{-iG}\left(\H_\rm{D}(N)\ot \F_\bb(w\H_\rm{T})\right),
\end{align}
via unitary transformation $U$ given by \eqref{U}.
Hence, the self-adjoint operator 
\begin{align}
U_G^{-1} H_\rm{phys} U_G,\quad U_G:=Ue^{-iG}
\end{align}
is to be considered as a gauge transformed Hamiltonian from the Lorenz gauge into the Coulomb gauge. 
Let $\rm{d}\Gamma^\rm{C}_\bb(\omega)$ be a second quantization operator of $\omega$ when regarded as 
an operator in $\H_\rm{T}$, that is, the multiplication operator of the matrix valued function
\begin{align}
\kk\mapsto \begin{pmatrix}
\omega(\kk) & 0 \\
0 & \omega(\kk)
\end{pmatrix}. 
\end{align}
\begin{Thm}
We have the operator equality 
\begin{align}
U_G^{-1}H_\rm{phys}U_G=\overline{H_D(V,N)+\rm{d}\Gamma^\rm{C}_\bb(\omega)+q\sum_{k=1,2,3}\int^\oplus_{\mathcal{X}^N} d\XX\, \alpha^{ak} A^{\rm{C}}_k(\xx^a)+E^\rm{C}},
\end{align}
where 
\begin{align}\label{equivalent-Coulomb}
A_k^\rm{C}(\xx)&:= \frac{1}{\sqrt{2}}\sum_{r=1,2}\left\{c^{r}\left(e_{rk}\frac{\hat{\chi^{\xx^a}_\rm{ph}}}{\sqrt{\omega}}\right)+
c^\dagger_{r}\left(e^r_{\;\;k}\frac{\hat{\chi^{\xx^a}_\rm{ph}}}{\sqrt{\omega}}\right)\right\},\quad k=1,2,3,\\
E^\rm{C}&:=-\frac{1}{2}q^2 \sum_{a,b}\alpha^{a\mu}\int^\oplus_{\mathcal{X}^N}d\bvec X \ip{k_\mu g^{\xx^a}_\nu}{g^{\xx^b \nu}}.
\end{align}
\end{Thm}
 \begin{proof}
 Let $\Phi\in \mathcal{D}_1\cap \left(\H_\rm{D}(N)\ot\F_\bb(w\H_\rm{T})\right)$. Then, 
 \begin{align}
 H_\rm{phys} U_G\Phi &= \tilde{H}[e^{-iG}\Phi] \no\\
 &=[H_\rm{DM}(V,N)e^{-iG}\Phi].
 \end{align}
 By Lemma \ref{transf}, this is the equivalence class to which the vector 
 \begin{align}
 e^{-iG}&\Bigg(H_0+q\sum_{a}\alpha^{a\mu}\int^\oplus_{\mathcal{X}^N}d\bvec X\frac{1}{\sqrt{2}}
\left[ c^\nu\left(Q_{\mu\nu}(\xx^a)\right)+c^\dagger_\nu\left(Q_\mu^{\;\;\nu}(\xx^a)\right)\right] \no\\
&\hskip 14mm -\frac{1}{2}q^2 \sum_{a,b}\alpha^{a\mu}\int^\oplus_{\mathcal{X}^N}d\bvec X \ip{k_\mu g^{\xx^a}_\nu}{g^{\xx^b \nu}}\Bigg)\Phi
 \end{align}
 belongs. If we apply $U^{-1}$ given in \eqref{U} and then $e^{-iG}$ to this equivalent class, the result is
  \begin{align}
 &\Bigg(H_0+q\sum_{a}\sum_{k=1,2,3}\alpha^{ak}\sum_{r=1,2}\int^\oplus_{\mathcal{X}^N}d\bvec X\frac{1}{\sqrt{2}}
\left[ c^r\left(Q_{kr}(\xx^a)\right)+c^\dagger_r\left(Q_k^{\;\;r}(\xx^a)\right)\right] \no\\
&\hskip 14mm -\frac{1}{2}q^2 \sum_{a,b}\alpha^{a\mu}\int^\oplus_{\mathcal{X}^N}d\bvec X \ip{k_\mu g^{\xx^a}_\nu}{g^{\xx^b \nu}}\Bigg)\Phi,
\end{align}
since by \eqref{Q-concrete} one sees that in the second term the summation over $\mu$ survives only for $\mu=1,2,3$ and 
summation over $\nu$ only for $\nu=1,2$.
But since the operator inside the bracket coincides with the one on the right hand side of \eqref{equivalent-Coulomb},
we obtain
\begin{align}\label{inc}
U_G^{-1}H_\rm{phys}U_G|_{\mathcal{D}_1\cap\left( \H_\rm{D}(N)\ot\F_\bb(w\H_\rm{T})\right)}\subset 
H_D(V,N)+\rm{d}\Gamma^\rm{C}_\bb(\omega)+q\sum_{k=1,2,3}\int^\oplus_{\mathcal{X}^N} d\XX\, \alpha^{ak} A^{\rm{C}}_k(\xx^a)+E^\rm{C}.
\end{align}
The operator on the right hand side of \eqref{inc} is clearly symmetric and the subspace $\mathcal{D}_1\cap \H_\rm{D}(N)\ot\F_\bb(w\H_\rm{T})$
 is a core of the self-adjoint operator $U_G^{-1}H_\rm{phys}U_G$ by Theorem \ref{ess-sa}. Hence, the assertion follows.
 \end{proof}
 At the last of the present subsection, we see that the ``residual" term $E^\rm{C}$ in \eqref{equivalent-Coulomb} certainly gives the Coulomb interaction energy
 between electrons in the limit where the ultraviolet cutoff is removed. We have to be careful to take the limit because the term $E^\rm{C}$ includes
 not only the Coulomb interaction between two different electrons but also electron's self-interaction energy through gauge field, which will diverge
 in the limit. This diverging self-interaction part should be adequately subtracted.
 
 Choose a special photon cutoff function whose Fourier transform is propotional to 
 $\chi_{\epsilon, \Lambda}$, by which we denote the characteristic function of the set 
 \begin{align}
 \{\kk\in\R^3\,|\, \epsilon\le |\kk| \le \Lambda \}.
 \end{align} 
 The parameter $\epsilon$ and $\Lambda$ represent the infrared and ultraviolet cutoffs, respectively.
 Then,
  \begin{Thm}
 Let $E^\rm{C}(\epsilon,\Lambda)$ be $E^\rm{C}$ when a photon cutoff function is taken so that $\hat{\chi_\rm{ph}}=\chi_{\epsilon,\Lambda}/(2\pi)^{3/2}$. 
 Then,
 \begin{align}\label{Coulomb}
 \lim_{\Lambda\to\infty}\lim_{\epsilon\to 0}\left(E^\rm{C}(\epsilon,\Lambda)+\frac{q^2}{8\pi^2}\Lambda\right)=-\frac{q^2}{4\pi}\sum_{a<b}\int^\oplus_{\mathcal{X}^N}d\XX\,\frac{1}{|\xx^a-\xx^b|},
 \end{align}
 where the limit is understood as the strong resolvent sense with operators on both sides being considered to be
 a self-adjoint operators.
 \begin{proof}
 By Lemma \ref{transf}, $E^\rm{C}(\epsilon,\Lambda)$ is a multiplication operator by a function
 \begin{align}
 \XX\mapsto -\frac{q^2}{2}\sum_{a,b}\alpha^{a\mu}\ip{k_\mu g^{\xx^a}_\nu}{g^{\xx^b\nu}}.
 \end{align}
 A direct computation shows that this function can be written as 
 \begin{align}
 -\frac{q^2}{2}\sum_{a,b}\alpha^{a\mu}\ip{k_\mu g^{\xx^a}_\nu}{g^{\xx^b\nu}}
 &= -\frac{q^2}{2}\cdot2 \sum_{a<b}\int\frac{d^3\kk}{(2\pi)^3}\frac{|\chi_{\epsilon,\Lambda}(\kk)|^2}{|\kk|^2}e^{i(\xx^a-\xx^b)\kk} -\frac{q^2}{2}\int\frac{d^3\kk}{(2\pi)^3}\frac{|\chi_{\epsilon,\Lambda}(\kk)|^2}{|\kk|^2}\no\\
 &=-\frac{q^2}{4\pi} \frac{2}{\pi}\sum_{a<b}\int_{\epsilon}^\Lambda dk\,\frac{\sin(|\xx^a-\xx^b | k)}{|\xx^a-\xx^b|k}-\frac{q^2}{8\pi^2}(\Lambda -\epsilon),
 \end{align}
 which converges uniformly in $\XX\in\mathcal{X}^N$ to 
 \begin{align}
 -\frac{q^2}{4\pi}\frac{2}{\pi} \sum_{a<b}\int_{0}^\Lambda dk\,\frac{\sin(|\xx^a-\xx^b | k)}{|\xx^a-\xx^b|k}-\frac{q^2}{8\pi^2}\Lambda,
 \end{align}
 as $\epsilon$ tends to zero. Hence, what we should show is that the multiplication operator 
 \begin{align}
E^\rm{C}_\rm{ren}(\Lambda):=-\frac{q^2}{4\pi} \frac{2}{\pi}\sum_{a<b} \int_{\mathcal{X}^N}^\oplus d\XX\, \int_0^\Lambda dk\,\frac{\sin(|\xx^a-\xx^b | k)}{|\xx^a-\xx^b|k} 
 \end{align}
 converges to the operator on the right hand side of \eqref{Coulomb}, which we call $E_\rm{ten}^\rm{C}$ hereafter, in the strong resolvent sense.
 
 To this end, take arbitrary $\Psi\in D(E_\rm{ren}^\rm{C})$. Since $D(E_\rm{ren}^\rm{C}) \subset D(E_\rm{ren}^{C}(\Lambda))$
 for all $\Lambda\ge 0$ as can be easily checked, it suffices to prove that 
 \begin{align}
 \norm{E_\rm{ren}^\rm{C}(\Lambda)\Psi - E_\rm{ren}^\rm{C}\Psi } \to 0
 \end{align}
 as $\Lambda\to\infty$. But this follows from the well known integral formula
 \begin{align}
 \lim_{L\to\infty}\int_0^L \frac{\sin y}{y}\,dy =\frac{\pi}{2},
 \end{align}
 as well as the estimate
 \begin{align}
 \norm{E_\rm{ren}^\rm{C}(\Lambda)\Psi - E_\rm{ren}^\rm{C}\Psi }&=\int_{\mathcal{X}^N} d\XX\,\norm{E_\rm{ren}^\rm{C}(\Lambda)(\XX)\Psi - E_\rm{ren}^\rm{C}(\XX)\Psi(\XX) }^2 \no\\
 &\le \frac{q^2}{4\pi}\sum_{a<b} \int_{\mathcal{X}^N} d\XX\,\norm{\left(\frac{2}{\pi}\int_0^{|\xx^a-\xx^b|\Lambda} \frac{\sin p}{p}\,dp - 1\right)\frac{1}{|\xx^a - \xx^b |}\Psi(\XX)}^2 \no\\
 &\to 0,
 \end{align}
as $\Lambda$ tends to infinity, justified by the Lebesgue dominant convergence theorem. 
 \end{proof}
 \end{Thm}
\subsection{Triviality of the physical subspace}
If $\hat{\chi_\rm{ph}}$ does not belong to the domain of $\omega^{-3/2}$, the definition of $G$ in \eqref{G}
makes no sense. In this case, the physical subspace is trivial: 
\begin{Thm}\label{trivial}
Suppose that $\hat{\chi_\rm{ph}}$ does not belong to $D(\omega^{-3/2})$. Then,
\begin{align}
V_\rm{phys}=\{0\}.
\end{align}
\end{Thm}
\begin{proof}
Let $\Psi\in D'$ and 
\begin{align}
\Omega^+(t,f)\Psi=0
\end{align}
for all $t\in\R$ and $f\in\mathfrak{h}\cap D((-\Delta)^{-1/4})$. Then, we have by \eqref{omega+} with $t=0$
\begin{align}
a_\mu\left(\frac{ik^\mu \hat{f^*}}{\sqrt{\omega}}\right)\Psi=\frac{iq}{\sqrt{2}}\sum_a\int_{\mathcal{X}^N}^{\oplus}d\bvec{X} \ip{\frac{\hat{\chi_\rm{ph}^{\xx^a}}}{\omega}}{f}\Psi(\XX).
\end{align}
Define $l:\mathfrak{h}\cap D((-\Delta)^{-1/4})\to \C$ by
\begin{align}
l(f)&:=\ip{\Psi}{a_\mu\left(\frac{ik^\mu \hat{f^*}}{\sqrt{\omega}}\right)\Psi} \no\\
&=\frac{iq}{\sqrt{2}}\sum_a\int_{\mathcal{X}^N}d\XX \ip{\Psi(\XX)}{ \ip{\frac{\hat{\chi_\rm{ph}^{\xx^a}}}{\omega}}{\hat{f}}\Psi(\XX)}.
\end{align}
We remark that, since 
\begin{align}
\int d\XX \norm{\Psi(X)}^2 \norm{ \frac{\hat{\chi_\rm{ph}^{\xx^a}}}{\omega} } <\infty,
\end{align}
the strong Bochner integral
\begin{align}
\int d\XX \norm{\Psi(X)}^2  \frac{\hat{\chi_\rm{ph}^{\xx^a}}}{\omega} \in L^2(\R^3)
\end{align}
is defined. Therefore, $l(f)$ can be rewritten as
\begin{align}\label{also}
l(f)= \ip{\frac{iq}{\sqrt{2}}\sum_a\int d\XX \norm{\Psi(\XX)}^2 \frac{\hat{\chi_\rm{ph}^{\xx^a}}}{\omega}}{\hat{f}}.
\end{align}\label{l-bound}
The definition of $l(f)$ and the estimate \eqref{annihilation-bound} implies there is a constant $C>0$ with
\begin{align}
| l(f) |\le C \norm{\Psi}\norm{N_\bb^{1/2}\Psi} \norm{\sqrt{\omega}\hat{f}}.
\end{align}
Take arbitrary $g\in D((-\Delta)^{-1/2})$. Then, $(-\Delta)^{-1/4}g \in  \mathfrak{h}\cap D((-\Delta)^{-1/4})$.
If we define $\tilde{l}: D((-\Delta)^{-1/2}))\to \C$ by
\begin{align}
\tilde{l}(g):=l((-\Delta)^{-1/4}g),
\end{align}
it follows by \eqref{l-bound} that 
\begin{align}
|\tilde{l}(g)| \le C_\Psi \norm {g} 
\end{align}
for some $\Psi$-dependent constant $C_\Psi$. From the Riesz Lemma, there is a $\Theta\in L^2(\R^3)$ such that
\begin{align}
\tilde{l}(g)&=\ip{\Theta}{\hat{g}} \no\\
&= \ip{\frac{iq}{\sqrt{2}}\sum_a\int d\XX \norm{\Psi(\XX)}^2 \frac{\hat{\chi_\rm{ph}^{\xx^a}}}{\omega}}{\frac{\hat{g}}{\sqrt{\omega}}},
\end{align}
where the second equality is obtained form \eqref{also}. This means that 
\begin{align}
\frac{iq}{\sqrt{2}}\sum_a\int d\XX \norm{\Psi(\XX)}^2 \frac{\hat{\chi_\rm{ph}^{\xx^a}}}{\omega}\in D(\omega^{-1/2}),
\end{align}
that is, the function
\begin{align}
\kk\mapsto \frac{1}{{\omega(\kk)}}\cdot \frac{iq}{\sqrt{2}}\sum_a\int d\XX \norm{\Psi(\XX)}^2 \frac{\hat{\chi_\rm{ph}^{\xx^a}}}{\omega}(\kk)
\end{align}
is square integrable. Hence, we find
\begin{align}\label{bounded-condition}
\frac{q^2}{2}\sum_{a,b}\int d^3\kk \frac{|\hat{\chi_\rm{ph}}(\kk)|^2}{\omega(\kk)^3} \phi^{a,b}(\kk) <\infty 
\end{align}
where
\begin{align}
\phi^{a,b}(\kk):=\int d\XX \int d\bvec Y \norm{\Psi(\XX)}^2\norm{\Psi(\bvec Y)}^2
e^{i\kk(\xx^a - \xx^b)}.
\end{align}

Suppose $\Psi\not=0$. Then 
\begin{align}
K:=\phi^{a,b}(\bvec 0) = \int d\XX \int d\bvec Y \norm{\Psi(\XX)}^2\norm{\Psi(\bvec Y)}^2 > 0.
\end{align}
By the Riemann-Lebesgue Theorem, $\phi^{a,b}$ is continuous in $\kk$, and thus there is an open ball $U\subset \R^3$
centered at the origin with
\begin{align}
\phi^{a,b}(\kk)\ge \frac{K}{2}, \q \kk \in U.
\end{align}
Since $\chi_\rm{ph}\in L^2(\R^3)$, one sees that 
\begin{align}
\left|\int_{\R^3\setminus U} d^3\kk \frac{|\hat{\chi_\rm{ph}}(\kk)|^2}{\omega(\kk)^3} \phi^{a,b}(\kk)\right| 
&\le \sup_{\kk\in\R^3\setminus U}  \left|\frac{\phi^{a,b}(\kk)}{\omega(\kk)^3}\right| \int_{\R^3\setminus U} d^3\kk |\hat{\chi_\rm{ph}}(\kk)|^2 <\infty.
\end{align}
Thus, \eqref{bounded-condition} says
\begin{align}\label{diverge}
\frac{q^2}{2}\sum_{a,b}\int_U d^3\kk \frac{|\hat{\chi_\rm{ph}}(\kk)|^2}{\omega(\kk)^3} \phi^{a,b}(\kk) <\infty .
\end{align}
But this is impossible when $\hat{\chi_\rm{ph}}\not\in D(\omega^{-3/2})$. To see this, 
note that
\begin{align}
\int_{\R^3\setminus U} \frac{|\hat{\chi_\rm{ph}}(\kk)|^2}{\omega(\kk)^3} \le 
\dfrac{1}{\inf_{\kk\in \R^3\setminus U}| \omega(\kk)|^3}\int_{\R^3\setminus U} |\hat{\chi_\rm{ph}}(\kk)|^2 <\infty,
\end{align}
since $U$ is centered at the origin and $\hat{\chi_\rm{ph}}$ belongs to $L^2(\R^3)$.
Thus, in the case where $\hat{\chi_\rm{ph}}\not\in D(\omega^{-3/2})$,
\begin{align}
\int_U d^3\kk \frac{|\hat{\chi_\rm{ph}}(\kk)|^2}{\omega(\kk)^3} =\infty.
\end{align}
Hence, the integration on the left hand side of \eqref{diverge} has to diverge in this case:
\begin{align}
\frac{q^2}{2}\sum_{a,b}\int_U d^3\kk \frac{|\hat{\chi_\rm{ph}}(\kk)|^2}{\omega(\kk)^3} \phi^{a,b}(\kk) \ge \frac{q^2}{2}\sum_{a,b}\int_U d^3\kk \frac{|\hat{\chi_\rm{ph}}(\kk)|^2}{\omega(\kk)^3}  \frac{K}{2}= \infty.
\end{align}
Therefore, $\hat{\chi_\rm{ph}}$ has to belong to $D(\omega^{-3/2})$ if $\Psi\not=0$. This comples the proof.
\end{proof}
\section*{Acknowlegdements}
The authors are grateful to Professor Asao Arai for valuable comments and discussions. 
This study is somewhat inspired by the lectures given by 
Assistant Professor Akito Suzuki at Hokkaido University. It is our pleasure to thank him.
\appendix
\section{Time-evolution of the position operator}\label{time-evolution-position}

In this section, we construct the time evolution of the position operator of the $a$-th Dirac particle, which
we denote by $\xx^a=(x^{a1},x^{a2},x^{a3})$. We will see that $\alpha^{aj}$ is the $j$-th component of the velocity operator 
(that is, the time derivative of the $j$-th component of the position operator) of the $a$-th Dirac particle.

For each $ r >0 $, we set
\begin{align}
\mathcal{M} _r := \{ f \in \H _\rm{D} \; \big| \; \supp f \subset \{ |\xx | \le r \} \} .
\end{align}

\begin{Lem}\label{L-Cone} For all $ r>0 $ and $ t \in \Real $,
\begin{align}
e^{-it H_\mathrm{D}} \mathcal{M} _r \subset \mathcal{M} _{r+ |t|} .
\end{align}
\end{Lem}

\begin{proof} See \cite[Section 1.5]{MR1219537}.
\end{proof}

For each $ r >0 $ and $ n \in \N $, we define the subspace $ \F _{r,n_0} \subset \F _\rm{DM} (N) $ as
\begin{align}
\F _{r,n_0} := ( \wg ^N \mathcal{M} _r ) \otimes \left(\op_{n=0}^{n_0} \ot _\rm{s} ^{n} \H _\rm{ph} \right) .
\end{align}

\begin{Thm} The followings hold:
\begin{enumerate}[(i)]
\item For all $ r >0 , \; n_0 \in \N $ and $ t\in \R $, $ \F _{r,n_0} \subset D(W(-t) x^{aj} W(t)) $. 

\item For all $ \Psi \in \cup _{r,n_0} \F _{r,n_0} \cap D(H_0) $, the mapping
\begin{align}
t \mapsto W(-t) x ^{aj} W(t) \Psi
\end{align}
is strongly continuously differentiable. Moreover, 
\begin{align}\label{WaW}
\frac{d}{dt} W(-t) x ^{aj} W(t) \Psi =  W(-t) \alpha ^{aj }W(t) \Psi 
\end{align}
holds.
\end{enumerate}
\end{Thm}

\begin{proof} 
\begin{enumerate}[(i)]
\item Using Lemma \ref{L-Cone} and the Trotter product formula \cite[Theorem VIII.31]{MR751959}, we see that $ e^{-itH_\rm{D} (V)} \mathcal{M} _r \subset \mathcal{M} _{r+ |t|} $, and thus $ e^{-itH_0} \mathcal{F} _{r,n_0} \subset \mathcal{F} _{r+ |t|,n_0}$. On the other hand, it is clear that $ H_1 \F _{r,n_0} \subset \F _{r,n_0+1} $. Hence we have
\begin{align}
& \sum _{m,n=0} ^\infty \norm{ e^{itH_0} U_m (-t, 0) x ^{a j} e^{-itH_0 } U_n (t, 0) \Psi } \no\\
& \le\sum _{m,n=0} ^\infty \frac{(C|t|)^{m+n}}{m!n!}(r+|t|+2n|t|)(n_0+m+n)^{1/2}\dots(n_0+1)^{1/2}\norm{\Psi} \no \\ 
&=\sum _{N=0}^\infty\sum_{n=0}^N \frac{(C|t|)^N}{N!} \begin{pmatrix} N \\ n
\end{pmatrix} (r+|t|+2n|t|) (n_0+N)^{1/2}\dots(n_0+1)^{1/2}\norm{\Psi} \no\\
& < \infty ,
\end{align}
for all $ \Psi \in \F _{r,n_0} $, which implies $ \Psi \in D(W(-t)x^{a j} W(t)) $.

\item By a direct calculation, we have
\begin{align}\label{w-derivative}
\Expect{(iH_\rm{DM}(V,N)) ^* \Psi , x ^{aj}  \Phi } - \Expect{x ^{aj} \Psi , iH_\rm{DM}(V,N) \Phi } = \ip{\Psi}{\alpha^{aj} \Phi } ,
\end{align}
for all $ \Psi , \Phi \in \cup _{r,n_0} \F _{r,n_0} \cap D(H_0) $. Using \eqref{w-derivative} and the fact that $\alpha^{aj}$ is bounded, we 
can mimic the proof of Theorem \ref{C_1-equation} to obtain \eqref{WaW}.
\end{enumerate}
\end{proof}

\bibliographystyle{junsrt}
\bibliography{myref_master}

\end{document}